\documentclass{amsart}
\usepackage{epsfig}
\usepackage{geometry} 
\usepackage{amsmath}
\usepackage{graphicx}
\usepackage{epsfig}
\usepackage{latexsym}
\usepackage{epsfig}
\usepackage[round]{natbib}
\usepackage{color}
\newtheorem{theorem}{Theorem}
\newtheorem{lemma}[theorem]{Lemma}
\newtheorem{corollary}[theorem]{Corollary}
\newtheorem{proposition}[theorem]{Proposition}
\newtheorem{example}[theorem]{Example}
\newtheorem{remark}[theorem]{Remark}
\theoremstyle{definition}
\newtheorem{definition}[theorem]{Definition}

\newenvironment{proofthmmodel}{\noindent {\textit{Proof of Theorem \ref{thm_model}.}}}{\quad \hfill $\Box$}

\newcommand{\PP}{{\mathbb P}}
\newcommand{\RR}{{\mathbb R}}

\newcommand{\cE}{\mathcal E}
\newcommand{\cM}{\mathcal M}

\newcommand{\x}{\mathtt{x}}

\newcommand{\tp}{\tilde{p}}
\newcommand{\tx}{\tilde{\x}}

\newcommand{\co}{\operatorname{co}}
\newcommand{\Inc}{\operatorname{Inc}}
\begin{document}

\title{Phylogenetic mixtures and linear invariants for equal input models}
\author{Marta Casanellas and Mike Steel}
\address{ MC: Department of Mathematics, Universitat Polit\`ecnica de Catalunya, Barcelona, Spain}
\address{ MS: Biomathematics Research Centre, University of Canterbury, Christchurch, New Zealand}

\begin{abstract}
The reconstruction of phylogenetic trees from molecular sequence data relies on modelling site substitutions by a
Markov process, or a mixture of such processes. In general, allowing mixed processes can result in different tree topologies becoming indistinguishable from the data, even for infinitely long sequences.
However, when the underlying Markov process supports  linear phylogenetic invariants, then provided these are sufficiently informative, the identifiability of the tree topology can be restored.
In this paper, we investigate a class of processes that support linear invariants once the stationary distribution is fixed, the `equal input model'.
This model generalizes  the `Felsenstein 1981' model (and thereby the Jukes--Cantor model) from four states to an arbitrary number of states (finite or infinite), and it can also be described by a `random cluster' process.
We describe the structure and dimension of the vector spaces of phylogenetic mixtures and of linear invariants for any fixed phylogenetic tree (and for all trees -- the so called `model invariants'), on any number $n$ of leaves. We also provide a precise description of the space of mixtures and linear invariants for the special case of $n=4$ leaves. By combining techniques from discrete random processes and  (multi-) linear algebra, our results build on a classic result that was  first established by James Lake in 1987.
\end{abstract}

\maketitle

\section{Introduction}
\label{intro}
Tree--based  Markov processes on a discrete state space play a central role in molecular systematics. These processes  allow biologists to model the evolution of characters and  thereby to develop techniques for inferring a phylogenetic tree for a group of species  from a sequence of characters (such as the sites at aligned DNA  or amino acid sequences  \citep{fel04}).  Under the assumption that each character evolves independently on the same underlying tree, according to a fixed Markov process, the tree topology  can be inferred in a statistically consistent way (i.e. with an accuracy approaching 1 as the number of characters grows) by methods such as maximum likelihood estimation (MLE) \citep{cha96} and techniques based on phylogenetic invariants \citep{fer15}.   This holds even though one may not know the values of the other (continuous) parameters associated with the model, which typically relate to the  length of the edges, and relative rates of different substitution types.

The assumption that all characters evolve under the same Markov process is a very strong one, and biologists generally allow the underlying process to vary in some way between the characters.  For example, a common strategy is to allow characters to evolve at different rates (i.e. the edge lengths are all scaled up or down in equal proportion at each site by a factor sampled randomly from some simple parameterized distribution). In that case, provided  the rate distribution is sufficiently constrained, the tree topology can still be inferred in a statistically consistent manner  \citep{all12,mat08}, and by using MLE, or related methods.

However,  when this distribution is not tightly constrained, or when edge lengths are free to vary in a more general fashion from character to character then different trees can lead to identical probability distributions on characters \citep{all12,ste94}.  In that case,  it can be impossible to decide which of two (or more) trees generated the given data, even when the number of characters tends to infinity. In statistical terminology, {\em identifiability} of the tree topology parameter is lost. For certain types of Markov models, however, identifiability of the tree topology is possible, even in these general settings. These
are models for which (i) linear relationships (called `linear phylogenetic invariants') exist between the probabilities of different characters, and which hold for all values of the other continuous parameters associated with the model (such as edge lengths) and (ii) these invariants can be used to determine
the tree topology (at least for $n=4$ leaves) \citep{ste11,stef07}.  The first such invariants, which we call \textit{linear topology invariants}, were discovered by James Lake in a landmark paper in 1987 \citep{lak87} for the  Kimura 2ST model, and the Jukes--Cantor submodel.

Linear topology invariants were known to exist for Kimura 2ST and Jukes--Cantor models, and the dimension of the corresponding (quotient) linear space had been computed for the Jukes--Cantor model in \citet{fu95} and \cite{ste95}. It is also known that
more general models such as Kimura 3ST or the general Markov model do not admit linear topology invariants (see for example \citep{sturmfels2004} and \citep{cas11}). Nevertheless, linear topology invariants had not been studied
for evolutionary models with more than 4 states or for models slightly more general than Jukes--Cantor.

In this paper we extend Lake--type invariants to a  more general setting and for another type of process, the `equal input'  model (defined shortly, but it can be regarded as the simplest  Markov process that allows  different states to have different stationary probabilities).
By building also on the approach of \citet{mat08} (which dealt just with  the 2-state setting) we investigate  the vector space of linear invariants,
and  describe the space of phylogenetic mixtures on a tree (or trees) under the equal input model once the stationary distribution is fixed. Note that the space of phylogenetic mixtures is dual to the space of phylogenetic invariants, and hence
studying one of these spaces translates into results for the other space. 
This leads to our main results  (Theorems 1 and 2)  which characterize the space of phylogenetic mixtures across all trees, and on a fixed tree (respectively), along with an algorithm for constructing a basis for the topology invariants.
It is worth pointing out that while linear \emph{topology} invariants are relevant for distinguishing distributions arising from mixtures of distributions on
particular tree topologies,
linear phlylogenetic invariants satisfied by distributions arising from mixtures of distributions on trees evolving under a particular model (\textit{model invariants})
can be used in model selection as in \citet{ked12}.  In brief summary, our  main results describe the vector
space (and its dimension) of the space of phylogenetic  mixtures of the equal input models for any  numbers $n$
of leaves and $\kappa$ of states:
\begin{itemize}
\item across all trees (Theorem~\ref{thm_model}) by providing a spanning set of independent points;
\item for a fixed tree (Theorem~\ref{thm_dimension}); and
\item for an infinite state version of the equal input model, known as Kimura's infinite allele model (Proposition~\ref{proinf}).
\end{itemize}

Using the duality between phylogenetic mixtures and linear invariants, in Corollary \ref{cor_topoinvar} we compute the dimension
of the quotient space of linear topology invariants and describe an algorithm for computing a basis of this space.
Note that the dimension of the space of mixtures had already been computed in \citet{cas12} and in \citet{fu95} for the Jukes--Cantor  model.  
Theorem~\ref{prop_linspace}  and Corollary~\ref{corocor} provide a more detailed description for trees with $n=4$ leaves. The case $n=4$ is of particular interest,
since the existence of a set of  linear phylogenetic invariants for this case and which, collectively, suffice to  identify the tree topology
means that there also exist  informative linear phylogenetic invariants that can identify any fully-resolved (binary) tree topology on   any number of leaves. This follows from the well-known fact that any binary tree topology is fully determined by its induced quartet trees  (for details and references, see  \citet{sem03}).

We also establish various other results along the way, including a  `separability
condition' from which a more general description of Lake--type invariants  follows (Proposition~\ref{lake}). We begin with some standard definitions, first for
Markov processes on trees, and then for the equal input model, which we show is formally equivalent to a random cluster
process on a tree (Proposition~\ref{ch07_equal_is_random}). We then develop a series of preliminary results and lemmas that will lead to the main results described above.

\section{Markov processes on trees}
\label{Markov_processes_on_trees}
Given a tree $T = (V,E)$ with leaf set $X$, a {\em Markov process on $T$} with state space $S$ is a collection of
random variables $(Y_v: v \in V)$ taking values in $S$, and which satisfies the following property.
For each interior vertex $v$ in $T$, if $V_1, \ldots, V_m$ are the sets of vertices in the connected components of $T-v$ then  the $m$ random variables $W_i = (Y_v: v\in V_i)$ are conditionally independent given $Y_v$.

Equivalently, if we were to direct all the edges away from some (root) vertex, $v_0$, then this condition says that
conditional on $Y_v$ (for an interior vertex $v$ of $T$) the states
in the subtrees descended from $v$ are independent of each other, and are also independent of the states
in the rest of the tree.

 A Markov process on $T$ is determined entirely by the
probability distribution $\pi$ at a root vertex $v_0$, and the assignment $e \mapsto  P^{(e)}$, that associates a transition matrix with each edge $e=(u,v)$ of $T$ (the edge is directed away from $v_0$).
Matrix  $P^{(e)}$ has row $\alpha$ and column $\beta$ entry equal to $P^{(e)}_{\alpha\beta} := \PP(Y(v)= \beta|Y(u)=\alpha)$, and so each row sums to 1.  If stochastic vector $\pi$ has the property that $\pi = \pi P^{(e)}$ for every edge $e$ of $T$, then $\pi$ is said to be a {\em stationary distribution} for the process.
A {\em phylogenetic model} is a Markov process on a tree where the transition matrices are required to
belong to a particular class $\cM$.

In this paper we will be concerned with trees in which the set $X$ of leaves are labelled, and all non-leaf (interior) vertices are unlabelled and have degree  at least three; these are called {\em phylogenetic $X$--trees} \citep{sem03}. A tree with a single interior vertex is called a {\em star}, while a tree for which every interior vertex has degree three is said to be {\em binary}. We will write $ab|cd$ for the binary tree on four leaves (a {\em quartet tree}) that has an edge separating leaves $a,b$ from $c,d$.     A function $\chi: X \rightarrow S$ is called a {\em character} and any Markov process on a tree with state space $S$ induces a (marginal) probability distribution on these characters.
An important algebraic feature of this distribution is that the probability of a character $\PP(\chi)$ under a Markov
process on $T$ is a polynomial function of the entries in the transition matrices.

\bigskip

\subsection{The equal input model}
\label{The_equal_input_model}
The {\em equal input model} ($EI$) for a set $S$ of $\kappa$ states is a particular type of Markov process on a tree, defined as follows.   Given a root vertex $v_0$ let  $\pi$ be a distribution of states at $v_0$ and
for each (directed) edge $e = (u,v)$ (directed away from $v_0$). In the $EI$ model, each transition matrix $P^{(e)}$ has the property that for
some value $\theta_e \in [0,1]$ and all states $\alpha, \beta \in S$ with $\alpha \neq \beta$ we have:
\begin{equation}
\label{Peq}
P^{(e)}_{\alpha\beta}= \pi_\beta \cdot \theta_e.
\end{equation}
 We shall assume that the distribution $\pi$ is strictly positive throughout the paper.

This model generalizes the familiar \emph{fully symmetric model} of $\kappa$ states (such as the `Jukes-Cantor model', when $\kappa=4$)
to allow each state to have its own stationary probability. In the case $\kappa =4$ with $S$ equal to the four nucleotide bases, the model is known as the {\em Felsenstein 1981 model}.   The defining property of the model is that the probability of a transition from $\alpha$ to $\beta$ (two distinct states) is the same, regardless of the initial state $\alpha (\neq \beta)$.

\begin{lemma}\label{lem_EI}
The following properties hold for the equal input model.
\begin{itemize}
\item[(i)] $P^{(e)}_{\alpha\alpha} = 1-\theta_e + \pi_\alpha \theta_e$.
\item[(ii)] $\pi$ is a stationary distribution for each vertex $v$ of the $T$ (i.e. $\PP(Y(v)=\alpha) = \pi_\alpha$).
\item[(iii)] The process is time-reversible (i.e. for each edge $e$, $\pi_\alpha P^{(e)}_{\alpha\beta} = \pi_\beta P^{(e)}_{\beta \alpha}$).
\item[(iv)] If $p$ is the probability that the ends of $e$ receive different states under the $EI$ model, then $p = (1- \sum_\alpha \pi_\alpha^2)\theta_e$.
\item[(v)] The process is multiplicatively closed. In other words, $(P^{(e)}P^{(e')})_{\alpha\beta} = \pi_\beta \theta$, where $\theta = 1-(1-\theta_e)(1-\theta_{e'})$.
\end{itemize}
\end{lemma}

\begin{proof}
For (i), $P^{(e)}_{\alpha\alpha}  = 1-\sum_{\beta \neq \alpha} P^{(e)}_{\alpha\beta}  = 1- \theta_e \sum_{\beta \neq \alpha} \pi_\beta = 1-\theta_e(1-\pi_\alpha)$.
For (ii), it suffices to show that if $(u,v)$ is a directed edge and $u$ has stationary distribution $\pi$ then $v$ does too.
But
$$\PP(Y(v)= \beta) = \sum_{\gamma}\pi_\gamma P^{(e)}_{\gamma\beta} = \pi_{\beta}P^{(e)}_{\beta\beta} + \sum_{\gamma \neq \beta} \pi_\gamma P^{(e)}_{\gamma\beta} = \pi_\beta.$$
For (iii), the result clearly holds if $\alpha = \beta$ so suppose $\alpha \neq \beta$. Then
$$\pi_\alpha P^{(e)}_{\alpha\beta} = \pi_\alpha(\pi_\beta\theta_e) = \pi_\beta (\pi_\alpha \theta_e) = \pi_{\beta}P^{(e)}_{\beta\alpha} .$$
For (iv),
$$p = \sum_\alpha \pi_\alpha \sum_{\beta \neq \alpha} P^{(e)}_{\alpha\beta} = \sum_\alpha  \pi_\alpha \sum_{\beta \neq \alpha} \pi_\beta \theta_e,$$
which simplifies for the expression in (iv).
Property (v) is left as an exercise.
\end{proof}

For an equal input model, the transition matrix $P^{(e)}$ has eigenvalue  $1-\theta_e$ with multiplicity $k-1$ (and eigenvalue 1 with multiplicity 1).
Also, for fixed $\pi$ the matrices $P^{(e)}$ commute, as they can be simultaneously diagonalized by a fixed matrix (which depends on $\pi$).
Equal input models with also have a continuous realisation with rate matrix $Q$ defined by its off-diagonal entries as follows:
$$Q _{\alpha\beta} = \pi_\beta, \mbox{ for all } \alpha, \beta\in S, \alpha \neq \beta$$
(the diagonal entries are determined by the requirement that each row of $Q$ sums to 0).
Then $P^{(e)} = \exp(Qt)$ for  $t= -\ln(1-\theta_e)$, and so $\theta_e = 1-e^{-t}$.
In the case where  $\pi$ is uniform, the $EI$ model reduces to the fully symmetric model in which all substitution events have equal probability.

One feature of the $EI$ model, that fails for most other Markov processes on trees, is the following. Let $\sigma$ be any partition of the state space $S$, and for a state $s \in S$ let $[s]$ denote the corresponding block of $\sigma$ containing $s$. Then for an $EI$ process $Y$ on the set $V$ of vertices of a phylogenetic tree $T$, let $\tilde{Y}$ be the induced stochastic process on $V$, defined by
$\tilde{Y}(v)=  [Y(v)]$ for all vertices $v$ of $T$.

\begin{proposition}
For any $EI$ model with parameters $\pi$ and $\{\theta_e\}$, and any partition $\sigma$ of $S$,
$\tilde{Y}$ is also an $EI$ Markov process on $T$, with parameters $\tilde{\pi}$ and $\{\theta_e\}$,
where for each block $B$ of $\sigma$, $\tilde{\pi}_B := \sum_{\beta \in B} \pi_\beta$.
\end{proposition}
\begin{proof}
By Theorem 6.3.2 of  \citet{kem76}, the condition for $\tilde{Y}$ to be a Markov process is that it satisfies a `lumpability' criterion that for any two choices $\alpha, \alpha' \in A \in \sigma$, and block $B \in \sigma$,
$$\PP(Y(v)  \in B|Y(u) = \alpha) = \PP(Y(v)  \in B|Y(u) = \alpha').$$
For each $B \neq A$, this last equality is clear from (\ref{Peq}), and since $\PP(Y(v)  \in A|Y(v) = \alpha) =1- \sum_{B \in \sigma, B \neq A}\PP(Y(v) \in B|Y(u) = \alpha)$ the criterion also holds
for the case $B=A$.
Finally, for $B \neq A$, $\PP(\tilde{Y}(v)=B|\tilde{Y}(u) = A) = \sum_{\beta \in B} (\pi_\beta \theta_e) = \tilde{\pi}_B\theta_e$.
\end{proof}

\bigskip

\subsection{A useful lemma}
\label{A_useful_lemma}
For results to come the following lemma, and its corollary  will be helpful.


\begin{lemma}\label{lem1}
For variables $x_1, x_2, \ldots, x_r$, consider polynomials $f_0({\bf x}),$ $\dots,$ $f_M({\bf x})$
 $\in \RR[x_1, \ldots, x_r]$ of the form
$$ f_i({\bf x}) = \sum_{A \subseteq [r]} c_A^{(i)} \prod_{j \in A} x_j, \quad c_A^{(i)}\in \RR.$$
\begin{itemize}
\item[(i)] Then $f_0\equiv 0$ (i.e. $c_A ^{(0)}= 0$ for all $A \subseteq [r]$) if and only if for any $t \neq 0$,
$f_0({\bf x}) =0$ for all ${\bf x} \in \{0, t\}^r.$
\item[(ii)]  Let $f=(f_1,\dots,f_M): \RR^r \rightarrow \RR^M$ and let $L:\RR^M\rightarrow \RR$ be a linear map.
Define an equivalence relation among the elements of $\{0, 1\}^r$ by $\x \sim \x'$ if $f(\x)=f(\x')$, and let $\x_1, \dots, \x_s$
be representatives of these equivalence classes. We call $q_i=f(\x_i)$, $i=1,\dots,s$.
Then $L(f(\x))=0$ for all $\x\in \RR^{r}$ if and only if $L(q_j)=0$ for  $j=1,\dots, s$.
\end{itemize}
\end{lemma}
\begin{proof}

\noindent (i) The `only if' part holds automatically; for the `if' direction, given any subset $B$ of $[r]$, let $h(B) = h({\bf x}^B)$ where
$x^B_i=t$ if $i \in B$ and $x^B_i=0$ otherwise. Then $h(B)=0$ by hypothesis, and $h(B) = \sum_{A\subseteq B} c_At^{|A|}$, by definition.
Applying the (generalized) principle of inclusion and exclusion it follows that, for each $A \subseteq [n]$,
$c_At^{|A|} = \sum_{B\subseteq A} (-1)^{|A-B|} h(B)=0$, so $c_A=0$.

\noindent (ii) The map $h=L\circ f$ satisfies the hypotheses of (i), hence $L( f(\x))=0$ for all $\x$ if and only if $L( f(\x))=0$
for all $\x \in \{0, 1\}^r.$ Then the statement follows immediately due to the definition of the equivalence relation.
\end{proof}

In what follows we will use this lemma to check linear relations among the character probabilities.

In the $EI$ model, once we fix $\pi$, the probability $\PP_T(\chi|\Theta)$  of observing a character at the leaves of $T$ satisfies
the hypotheses of the corollary with  $r$ equal to the number of edges and variables in $\Theta=\{\theta_e\}_{e\in E(T)}$.
Indeed, by Lemma \ref{lem_EI} (i), any entry of the transition matrix $P^{(e)}$ is a linear function of $\theta_e$ and hence the expression
\begin{equation}
\label{ppsum}
\PP_T(\chi|\Theta) =\sum_{(s_v)_v \in S^{{\rm Int}(T)}}
\pi_{s_{v_0}}\prod_{(u,w)\in E(T)}
P^{(e)}_{s_u,s_w}
\end{equation}
(where the sum is over the states at the set ${\rm Int}(T)$ of interior vertices of $T$ and subject to the
convention that $s_w=\chi(l)$ if $w$ is the leaf $l$) satisfies the hypotheses of Lemma \ref{lem1}.

\begin{remark}\rm Lemma \ref{lem1} can be be slightly modified to accommodate  substitution matrices with more parameters as it was done in \citet{fu95}.
\end{remark}

\bigskip

\subsection{The equal input model as a random cluster model}
\label{The_equal_input_model_as_a_random_cluster_model}
Our alternative description of the $EI$ model is as an instance of the (finite) {\em random cluster model} (briefly \emph{RC}) on trees
(this phrase is also used  to study processes on graphs, such as the `Ising model' in physics).
For an unrooted phylogenetic tree with leaf set $[n]$, each edge $e$ of $T$ is cut independently with probability $p_e$.
The leaves in each connected component of the resulting disconnected graph  are then all assigned the same state $s$ with probability $\pi_s$,
independently of assignments to the other components (see Fig.~\ref{fig_randomcluster}). More precisely, for any binary function
$g:E(T)\rightarrow \{0,1\}$, define $C(g)$ to be the set of connected components in $T \setminus \{e\in E(T) |g(e)=1\}.$
Then the probability $\PP_T(\chi|\{p_e\}_e)$ of observing a character $\chi$ at the leaves of $T$ under the
$RC$ model is
\begin{equation}
\label{param$RC$}
\sum_{g:E(T)\rightarrow \{0,1\}} \PP(\chi|g)p_e^{g(e)}(1-p_e)^{1-g(e)}
\end{equation}
where $\PP(\chi|g)$ is 0 if $\chi(i)\neq \chi(j)$ for some leaves $i,j$ in the same connected component in $C(g)$ and is equal to
$\prod_{c \in C(g)}\pi_{\chi_c}$ otherwise (where $\chi_c$ denotes the value of $\chi$ at the leaves of $T$ that are in $c$).
In particular, the $RC$ model also satisfies the hypotheses of Lemma \ref{lem1}.

\begin{figure}[htb]
\centering
\includegraphics[scale=0.75]{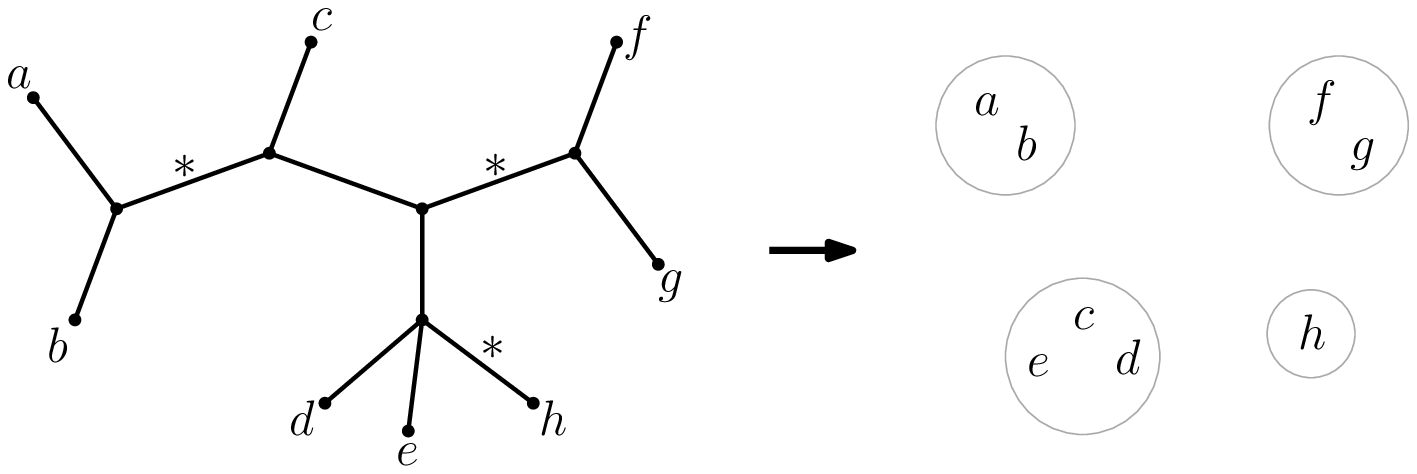}
\caption{Cutting the three edges marked * in the tree on the left leads to the partition of $X$ shown at right. Under the random cluster model
these four blocks are independently assigned states from the distribution $\pi$.}
\label{fig_randomcluster}
\end{figure}

\begin{proposition}
\label{ch07_equal_is_random}
The $EI$ model with parameters $\pi$ and $\{\theta_e\}$ produces an identical probability distribution on characters as the
random cluster model in which $p_e=\theta_e$ for each edge $e$ of $T$.
\end{proposition}

\begin{proof}
For the two models the probability  of a given character (given by Eqns. (\ref{ppsum}) and (\ref{param$RC$}))
satisfies the conditions required  by Lemma \ref{lem1} (ii), and so we can use
it with $M=2$ and $L$ the difference between the probability of a given character by the two models. Therefore,
it suffices to show that the two models produce the same probability distribution on characters  whenever  $\theta_e =1$ for all $e \in F$ and $\theta_e =0$ of all edges $e$ of $T$ not in $F$ (for all possible choices of subset $F\in E(T)$).
Given $F$, notice that if $\theta_e=1$ for a directed edge $e=(u,v)$ of $T$ in the $EI$ model, then
$P^{(e)}_{\alpha\beta} = \pi_\beta$ for \underline{all} $\beta \in S$, including $\beta=\alpha$. In other words, when $\theta_e=1$ for $e=(u,v)$, the state at $v$ is completely independent of the state at $u$. This is equivalent to cutting the edge and assigning a random state according to the distribution $\pi$ to $v$, and thereby to all the other vertices of $T$ for which there is a path to $v$ that does not cross another edge in $F$ (since $P^{(e)}$ is the identity matrix on those edges); this is just the process described by the random cluster model.
\end{proof}

\bigskip

\section{Linear phylogenetic invariants in phylogenetics}
\label{Linear_phylogenetic_invariants_in_phylogenetics}

\begin{definition}
Consider a phylogenetic model $\mathcal{M}$ with state space $S$ on a phylogenetic tree $T$ with $n$ leaves.
A \emph{phylogenetic invariant} of a tree $T$ under the model $\mathcal{M}$ is a polynomial $f$ in $S^n$ indeterminates
that vanishes on any distribution $\PP_{T,\Theta}$
that arises under the phylogenetic model $\mathcal{M}$ (that is, $f(p)=0$ if $p=\PP_{T,\Theta}$, for any set $\Theta$
of transition matrices and distribution at the root vertex).

We say that a polynomial in $S^n$ coordinates is a \emph{model invariant} if it is a phylogenetic invariant for any tree on $n$ leaves
under the phylogenetic model $\mathcal{M}$. A phylogenetic invariant of a tree $T$ that is not a model invariant is called a
\emph{topology invariant}.
\end{definition}

A phylogenetic invariant is a \emph{linear phylogenetic invariant}
(resp. \emph{linear model invariant}, \emph{linear topology invariant}) if each monomial involves exactly one indeterminate and has degree 1.
Note that this implies that the polynomial is homogeneous (the independent term is 0). There are phylogenetic invariants of degree 1 that are not homogeneous,
for example the \emph{trivial} phylogenetic invariant that arises from the observation that in a distribution all
coordinates must sum to one. However, taking this trivial invariant into account, any other phylogenetic invariant of degree 1 can be rewritten
as a homogeneous phylogenetic invariant of degree 1 (indeed, $\sum_i a_i \x_i+a$ is a phylogenetic invariant if and only if $\sum_i (a_i+a) \x_i$ is a phylogenetic invariant).
This is why we only call \emph{linear} phylogenetic invariants those that are homogeneous of degree 1. The sets of  linear model invariants and
linear phylogenetic invariants of a tree $T$ are vector spaces.

Linear phylogenetic invariants are of particular interest since they hold even if the process changes from
character to character (provide it stays within the model for which the invariant is valid).
An important early example of linear phylogenetic invariants were discovered by James Lake in 1987 \citep{lak87}.
In this paper, we first provide a new and more general version of Lake's invariants. It is the first time that linear topology invariants are given for non-uniform stationary distributions and for models on any number of states,
provided that they satisfy what we call the Partial Separability condition (see below). It is worth noting that in our Lake--type invariants the stationary distribution is assumed to be known.

For any phylogenetic $X$--tree, $T$ (not necessarily binary), and an interior vertex $v$ of $T$ consider the disconnected graph $T-v$.  Let $t$ and $t'$ be two of the trees incident with $v$.

Suppose that a Markov process $Y$ on  $T$ takes values in state space $S$. For any state $s$ of $S$
 write $Y(t)=s$ if all the leaves of $T$ that are in $t$ are in state $s$ (similarly for $t'$).
Consider the following property.

\begin{itemize}
\item[(PS)] {\em Partial separability.}
For some  interior vertex $v$, and for some subset $\{a_1, a_2,$ $b_1, b_2\}$ of four distinct elements of $S$ one has
$$\PP(Y(t)=a_i|Y(v)=s) = \pi(a_i)c, \mbox{ when } s \in S-\{a_1, a_2\}, \, i=1,2;$$
and
$$\PP(Y(t')=b_j|Y(v)=s) = \pi(b_j)d, \mbox{ when }  s \in S-\{b_1, b_2\}, \, j=1,2.$$

\end{itemize}
Here $c$ and $d$ are arbitrary functions dependent on the tree and associated parameters (but not
the states) and $\pi$ is an arbitrary function of the states such that $\pi(a_i)\neq 0$, $\pi(b_i)\neq 0$, $i=1,2$
(for various models with $\pi$ given by the stationary distribution).

Partial separability is satisfied by various models. For example, when  $|S|=4$,
it holds for the Kimura 2-ST model (and hence the Jukes-Cantor model) by taking $\{a_1, a_2\} = \{A, G\}$ (purines) and
$\{b_1, b_2\} = \{C, T\}$ (pyrimidines), in which case $\pi(a_i)=\pi(b_i) = \frac{1}{4}$ for $i=1,2$.
The property also holds for the fully symmetric model on any number of states. Moreover, the property holds for the $EI/RC$ model
on any number of states if $t$ and $t'$ are single leaves.
The partial separability condition should be viewed as an algebraic constraint rather than as a natural condition that one might expect to hold for  most evolutionary models. For instance it, is not a natural property satisfied by evolutionary models and, for instance, it is not satisfied for the $EI/RC$ model if $t$ or $t'$ are not single leaves.

Let $\cE$ be \emph{any} event that involves the states at  the leaves of $T$ not in $t$ or $t'$.
For example, if $y$ and $y'$ are leaves of $T$ not in $t$ or $t'$ then  $\cE$ might be the event that
$Y(y)=s$ and $Y(y')  = s'$ for some particular $s, s'\in S$.

Let us write  $p_{\cE ij}$ for the probability of the three--way conjunction $\cE\wedge \{Y(t)=a_i\} \wedge\{ Y(t')=b_j\}$.
Notice that $p_{\cE ij}$ is a sum of probabilities of various characters (i.e. a marginal distribution).
Let $$\tp_{\cE ij} = \frac{1}{\pi(a_i)\pi(b_j)} \cdot p_{\cE ij} \mbox{ and  let }
\Delta:= \tp_{\cE 11}+ \tp_{\cE 22} - \tp_{\cE 12}- \tp_{\cE 21}. $$

\begin{proposition}[Lake--type invariants]
\label{lake}
If a Markov process on $T$ satisfies the partial separability condition (PS), then
$\Delta=0.$
\end{proposition}
\begin{proof}
By the Markov property,
$$p_{\cE ij} = \sum_{s}\PP(Y(v)=s)\cdot  \PP(\cE|Y(v)=s)\cdot \PP(Y(t)={a_i}|Y(v)=s)\cdot \PP(Y(t')={ b_j}|Y(v)=s).$$
Let
$r_{ij} = \pi(a_i)\cdot \pi(b_j)$, and let
$$\Delta_s =  r_{22}p_1p'_1 + r_{11}p_2p'_2 - r_{12}p_2p'_1 - r_{21}p_1p'_2,$$
where $p_i = \PP(Y(t)={a_i}|Y(v)=s)$, and $p'_j = \PP(Y(t')={ b_j}|Y(v)=s)$.
Then we can write
$$\Delta ={ \frac{1}{\pi(a_1)\pi(a_2)\pi(b_1)\pi(b_2)}}\sum_s \PP(Y(v)=s) \cdot \PP(\cE|Y(v)=s) \cdot \Delta_s.$$
Thus it suffices to show that $\Delta_s=0$ for all $s$.

We consider three cases:
(i): $s\in \{a_1, a_2\}$, (ii) $s \in \{b_1, b_2\}$ and (iii) $s \in S-\{a_1, a_2, b_1, b_2\}$.

In Case (i), suppose $s=a_i$. Then
$p'_1 = \pi(b_1)d$ and $p'_2 = \pi(b_2)d$, and so
$$\Delta_s = d[p_1r_{22}\pi(b_1) + p_2r_{11}\pi(b_2) - p_2r_{12}\pi(b_1) - p_1 r_{21}\pi(b_2))].$$
$$= dp_1[r_{22}\pi(b_1)- r_{21}\pi(b_2))] + dp_2[r_{11}\pi(b_2) -r_{12}\pi(b_1)] = 0+0 = 0.$$

Case (ii) is similar.  In Case (iii), $p_ip_j' = r_{ij}cd$ and so $$\Delta_s = cd [r_{22}r_{11}+ r_{11}r_{22} -r_{12}r_{21} - r_{21}r_{12}] = 0.$$
\end{proof}

\begin{example}\label{ex_lake}
\rm
When we take $t$ and $t'$ single leaves, the  $EI/RC$ model satisfies the (PS) property and Lemma~\ref{lake} can be applied.
If the stationary distribution $\pi$ is fixed, then $\Delta$ gives rise to two types of linear phylogenetic invariants for the quartet tree $12|34$,
$$H_1: \quad \frac{\x_{xyxy}}{\pi(x)\pi(y)}+\frac{\x_{xyzw}}{\pi(z)\pi(w)}-\frac{\x_{xyzy}}{\pi(z)\pi(y)}-\frac{\x_{xyxw}}{\pi(x)\pi(w)}$$
$$H_2: \quad \frac{\x_{xyyx}}{\pi(x)\pi(y)}+\frac{\x_{xywz}}{\pi(z)\pi(w)}-\frac{\x_{xyyz}}{\pi(z)\pi(y)}-\frac{\x_{xywx}}{\pi(x)\pi(y)}$$
(here $\x_{\chi_1\chi_2\chi_3\chi_4}$ is the coordinate that corresponds to $\PP_T(\chi_1\chi_2\chi_3\chi_4)$).
To see how these follow from Proposition~\ref{lake}, for $H_1$ take $x=a_1, y=b_1, z=a_2, w=b_2$ and let $\cE$ be the event that $Y(1) = a_1$
and $Y(2)=b_1$; for  $H_2$ take  $x=b_1, y=a_1, z=b_2, w=a_2$ and let $\cE$ be the event that $Y(1) = b_1$ and $Y(2)=a_1$.
Note that these are topology invariants because
the first is not a phylogenetic invariant for the quartet $13|24$ while the second is not a phylogenetic invariant for $14|23$.
\end{example}

\section{Generating linear invariants for the $RC/EI$ model on $\kappa$ states}
\label{Generating_linear_invariants}

\subsection{Combinatorial concepts and terminology}
\label{Combinatorial_concepts_and_terminology}
Let $T$ be a phylogenetic $X$--tree, $X=[n]$, and consider a Random Cluster model (or Equal Input model) on $T$,
with stationary distribution  $\pi$ on a set $S$ of $\kappa$ states.
Henceforth we assume that $\pi$ is fixed and it is positive, that is, $\pi_s \neq 0$  $\forall s \in S$.
We denote by $e_i$ the pendant edge incident with leaf $i$. 
A character $\chi:[n]\rightarrow S$ shall be denoted as $\chi=\chi_1\dots\chi_n$ if $\chi_i=\chi(i)$ for $i=1,\dots,n$.
We let $Ch(n,\kappa)$ to be the set of characters on $[n]$ for a fixed state space ($S$) of size $\kappa$ and denote by $N$ its cardinality ($N=\kappa^n$).
We  think of a distribution $\PP_{T,\Theta}$ on the set of characters under the $RC$ model  on $T$ as a vector of $Ch(n,\kappa)$ coordinates and therefore lying in the real vector
space with coordinates $\x_{\chi}$, $\chi\in Ch(n,\kappa)$ (the point $\PP_{T,\theta}$ has coordinate $\x_{\chi}$ equal to $\PP_{T}(\chi|\Theta)$).

Let $F$ be a \emph{subforest} of  $T$, that is, a subgraph comprised of a collection of vertex disjoint trees $\{T_1,\dots,T_r\}$ such that the only nodes of degree
$\leq 1$ in  $T_i$ are leaves of $T$ (we allow $T_i$ to be formed by only one leaf and we allow $F=\{T\}$ also).
We say that a subforest $F=\{T_1,\dots,T_r\}$ is a \emph{full} subforest if $\cup_i\mathcal{L}(T_i)=X$; we let  $\mathcal{F}_T$ be the set of
full subforests of $T$. For a full subforest $F$, we define
$\Theta_F$ to be the following collection of edge parameters under the $RC$ model: $\theta_{e}=0$ if $e\in E(T_i)$ for some $T_i \in F$ and $\theta_e=1$ for all other edges $e$.
We denote by $\sigma(F)$ the partition that $F$ describes on $[n]$, that is, two leaves are in the same block of $\sigma(F)$ if they lie in the
same subtree of $F$.
The full subforest formed by singletons will be called the \emph{trivial} subforest.

Given a character $\chi$, we define $\sigma(\chi)$ to be the partition $\{S_1,\dots ,S_l\}$ of the set of leaves defined according to
``two leaves $i,j$ are in the same block of the partition if $\chi_i=\chi_j$". Note that given a full subforest $F=\{T_1,\dots, T_r\}$ of $T$ and a character
$\chi$, $\PP_T(\chi|\Theta_F)$ is zero if $\sigma(F)$ does not refine $\sigma(\chi)$ and is equal to
$\prod_{i=1}^r \pi_{s_i}$ otherwise (here $s_i$ stands for the value of $\chi$ at the leaves of $T_i$).

For any partition $\sigma$ of $[n]$, and any phylogenetic tree $T$ on $[n]$, we say that
$\sigma$ is {\em convex} on $T$ (or {\em compatible} with $T$)  if the collection of induced subtrees
$\{T[B]: B \in \sigma\}$ are vertex disjoint  (here $T[B]$ is the minimal connected subgraph (subtree) of $T$ containing the leaves in $B$).
Let ${\rm co}(T)$ be the set of partitions of $[n]$ that are convex on $T$.
There is a natural correspondence between full subforests of $T$ and convex partitions on $T$ that associates to each partition $\sigma\in \co(T)$
the full subforest $F_T(\sigma)=\{T[B]: B \in \sigma\}$. Therefore, the number of full subforests of a tree $T$ is equal to $|\co(T)|$, $|\mathcal{F}_T|=|\co(T)|$.
When $T$ is a binary tree, $|{\rm co}(T)|= F_{2n-1}$
where $F_k$ is the $k$--th Fibonacci number, starting with $F_1=F_2=1$ (see \citet{ste95}). By contrast, for a star tree on $[n]$ we have
$|{\rm co}(T)| = 2^n - n$.  A   partition $\sigma=\{B_1, \ldots, B_k\}$ of $[n]$ is
{\em incompatible}  with $T$ if it is not convex on $T$, that is, there exist two blocks $B_i$ and $B_j$ from $\sigma$ for which $T[B_i]$ and $T[B_j]$
share at least one vertex.
A {\em singleton block} $B$ of $\sigma$ is a block of size 1. The number of partitions of $[n]$ is known as the \emph{Bell number} $B_n$.

Finally, let ${\rm Inc}(T)$ be the set of partitions of $[n]$ that are not convex on $T$ (i.e. they are `incompatible' with $T$).
Thus $|{\rm Inc}(T)| = B_n - |{\rm co}(T)|$.

\subsection{Results}
\label{Results}

\begin{lemma}\label{lem_forests} (a) Let $\Theta$ be a collection of parameters $(\theta_e)_{e\in E(T)}$ such that $\theta_e$ is either
$0$ or $1$ for all $e\in E(T)$. Then there exists a unique full subforest $F\in \mathcal{F}_T$ such that  $\PP_{T,\Theta}=\PP_{T,\Theta_F}$.

(b) A degree 1 polynomial $\sum_{\chi}\lambda_{\chi}\x_{\chi}$ is a linear phylogenetic invariant for a tree $T$  if and only if
$$\sum_{\chi}\lambda_{\chi}\PP_T(\chi|\Theta_F)=0$$
for any full subforest $F\in \mathcal{F}_T$.
\end{lemma}

\begin{proof}
(a) We first prove that two full subforests $F$ and $G$ satisfy  $ \PP(\chi|\Theta_G)\neq  \PP(\chi|\Theta_F)$ for some $\chi$ if $F\neq G$.
As $F,G$ are full subforests, they are different if and only if they induce different partitions $\sigma(F)$, $\sigma(G)$ on the set of leaves.
Then there exists an  edge $e_0$ such that $e_0$ is compatible with $\sigma(F)$ (i.e, $\sigma(F)$ refines the bipartition induced by $e_0$) but is not compatible with $\sigma( G)$ (or the other way around).
If $\chi$ is the character that assigns state $x$ at the leaves of one connected component of $T-e_0$ and state $y\neq x$ at the leaves
of the other component, then $\PP(\chi|\Theta_G)=0$ while $\PP(\chi|\Theta_F)$ is not zero.


Given $\Theta$, let $A$ be the set of edges $e$ in $T$ such that $\theta_e=1$. Let $\sigma(T\setminus A)$ be the partition induced on $X$ when
removing all edges in $A$ (if an edge in $A$ is a pendant edge, then removing it means that we separate the corresponding leaf).
If $F$ is the subforest $F_T(\sigma(T\setminus A))$, then we have $\PP_{T,\Theta}=\PP_{T,\Theta_F}.$

(b) This follows from part (a) and Lemma \ref{lem1} (ii).
\end{proof}

Let $\Theta$  be a collection of edge parameters on a tree $T$ evolving under the $RC$ model. For a site character $\chi$, we define
$$\tp_{\chi}^T(\Theta)=\frac{\PP_T(\chi|\Theta)}{\pi_{\chi_1}\pi_{\chi_2}\dots\pi_{\chi_n}}.$$

We call $\tilde{\x}_{\chi}$ the corresponding coordinates: $\tilde{\x}_{\chi}=\frac{\x_{\chi}}{\pi_{\chi_1}\pi_{\chi_2}\dots\pi_{\chi_n}}.$

\begin{lemma}\label{lem_equal}  We say that two characters $\chi$ and $\chi'$ are \emph{equivalent}, $\chi \equiv \chi'$, if $\sigma(\chi)=\sigma(\chi')$ and $\chi_i=\chi'_i$ for any leaf $i$ that belongs to a block of the partition of cardinality greater than or equal to 2.
Let $\chi$, $\chi'$ be two characters on the set $X=[n]$.
\begin{enumerate}
\item[(a)] If $\chi \equiv \chi'$ then 
$\tilde{\x}_{\chi}-\tilde{\x}_{\chi'}$ is a linear model invariant.
\item[(b)] If $\pi$ is not invariant by any permutation of the set of states, then for any tree $T$ the equality $\tp_{\chi}^T(\Theta)=
\tp_{\chi'}^T(\Theta)$ for every $\Theta$ implies that $\chi \equiv \chi'$ (i.e. in this case every linear phylogenetic invariant of
type $\tilde{\x}_{\chi}-\tilde{\x}_{\chi'}$ satisfies $\chi \equiv \chi'$). 
\end{enumerate}
\end{lemma}

\begin{proof} (a) Let $\chi$ and $\chi'$ be two equivalent characters, let $\sigma$ be $\sigma(\chi)=\sigma(\chi')$, and let $T$ be any $X$--tree.
According to Lemma \ref{lem_forests} (b) we need to check that
$\tp_{\chi}(\Theta_F)=\tp_{\chi'}(\Theta_F)$ for any $F=\{T_1,\dots,T_r\}\in \mathcal{F}_T$.

If $\sigma(F)$ does not refine $\sigma$, then $\PP_T(\chi|\Theta_F)$ and $\PP_T(\chi'|\Theta_F)$ are zero and we are done.

If $\sigma(F)$ does refine $\sigma$, then $\PP_T(\chi|\Theta_F)=\pi_{s_1}\dots\pi_{s_r}$ where $s_i$ is the value of $\chi$ at the leaves of $T_i$ (note that we may have $s_i=s_j$). Therefore $\tp_{\chi}^T(\Theta_F)=\frac{1}{\pi_{s_1}^{n_1-1}\dots\pi_{s_r}^{n_r-1}}$ where $n_i=|\mathcal{L}(T_i)|$. As $\sigma(F)$ refines $\sigma(\chi)=\sigma(\chi')$ and the states of $\chi$ and $\chi'$ coincide for any block of $\sigma$ of size $\geq 2$, the states of $\chi$ and $\chi'$ also coincide at the  leaves of $T_i$ if $n_i\geq2$. Therefore, $\tp_{\chi}^T(\Theta_F)=\tp_{\chi'}^T(\Theta_F)$.

As for (b), assume that $\pi$ is not invariant by any permutation of the set of states (i.e. $\pi_s=\pi_t$ if and only if $s=t$). Assume that for a tree $T$ we have $\tp_{\chi}^T(\Theta_T)=\tp_{\chi'}^T(\Theta_T)$ for any  collection of edge parameters $\Theta_T$. Then, for each block $B_i$ of $\sigma(\chi)$ of size $b_i$ greater or equal than 2  consider the forest $F_i=\{T_{B_i},\cup_{l \notin B_i}\{l\}\}$, where $T_{B_i}$ is the smallest subtree of $T$ joining the leaves in $B_i$. Then $\tp_{\chi}^T(\Theta_{F_i})=\frac{1}{\pi_{s_i}^{b_i-1}}$ if $s_i$ is the state of $\chi$ at the leaves of $B_i$. By hypothesis this is equal to $\tp_{\chi'}^T(\Theta_{F_i})$. But $\tp_{\chi'}^T(\Theta_{F_i})$ is zero if $\sigma(\chi')$ does not contain the block $B_i$. Performing the same argument for any block $B_i$ of size $b_i\geq 2$ we obtain  $\sigma(\chi)=\sigma(\chi')$. Now for each such block $B_i$ we have $\tp_{\chi}^T(\Theta_{F_i})=\tp_{\chi'}^T(\Theta_{F_i})$ and hence $\frac{1}{\pi_{s_i}^{b_i-1}}=\frac{1}{\pi_{s'_i}^{b_i-1}}$ if $s'_i$ is the state of $\chi'$ at the leaves of $B_i$. As $b_i\geq 2$, the assumption on $\pi$ implies $s_i=s'_i$. Thus, $\chi$ and $\chi'$ are equivalent characters.
\end{proof}

\begin{remark}\label{rem_symmetric}\rm
If $\pi$ is the uniform distribution  (i.e we consider the $\kappa$-state fully symmetric model), then we have
$\PP_T(\chi |\Theta)=\PP_T(\chi' |\Theta)$ if and only if $\sigma(\chi)=\sigma(\chi')$. Indeed, in this case
if we consider any permutation $g$ of the set of of states $S$, the polynomials
$\x_{\chi}-\x_{g\cdot\chi}$ are linear phylogenetic invariants for any tree (see \citet{cas12}),
where $g\cdot\chi$ stands for the corresponding permutation of states at the leaves. But these polynomials
can also be rewritten as $\x_{\chi}-\x_{\chi'}$ for $\sigma(\chi)=\sigma(\chi')$.
\end{remark}

\bigskip

\noindent
{\bf Examples: $n=3$ and $n=4$}
\begin{itemize}
\item
For $n=3$, Lemma~\ref{lem_equal}  gives the following. If $\kappa \geq 3$ and we consider three different states $x,y,z$ and another set of three different states $x',y',z'$, the linear invariants obtained in Lemma~\ref{lem_equal} are:
 $$ \tx_{xyz} - \tx_{x'y'z'},\, \tx_{xxy}  - \tx_{xxz},\, \tx_{xyx} - \tx_{xzx},\, \tx_{yxx}  - \tx_{zxx}.$$
\item
For $n=4$, Lemma~\ref{lem_equal}  gives the following.
If $\kappa \geq 4$ and we consider four different states $x,y,z,w$ and another set of four different states $x',y',z',w'$, the linear phylogenetic invariants of Lemma~\ref{lem_equal} are:
$$ \tx_{xyzw}- \tx_{x'y'z'w'},\,  \tx_{xxyz}  - \tx_{xxy'z'}, \, \tx_{xxxy}  - \tx_{xxxy'}, $$
and the analogous invariants obtained for the other partitions of $[4]$ involving singletons.
\end{itemize}
\hfill$\Box$

\bigskip

There are several ways to construct linear invariants from smaller trees and a systematic way to find model invariants for certain models with stationary distribution has been described in \citet{fu91}.
The most immediate one, used already in the quoted paper, uses the following marginalization lemma.
If $i$ is a leaf of $T$, we call $T_i$ the tree obtained by removing leaf $i$ and its incident edge, and suppressing the resulting degree--2 vertex if the interior node had degree 3.

\begin{lemma}\label{lem_margin}  Let $i$ be a leaf of a phylogenetic $[n]$--tree $T$ and let $T_i$ be the corresponding tree.
Let $l$ be a linear homogeneous map $l:\RR^{\kappa^ {n-1}}\rightarrow \RR$ and let $M_i:\RR^{\kappa^ {n}}\rightarrow \RR^{\kappa^ {n-1}}$ be the marginalization map at leaf $i$. If $l(p_i)=0$ for any distribution $p_i$ from a Markov process on the tree $T_i$, then $(l\circ M_i)(p) =0$ for any distribution $p$ that comes from a Markov process on the tree $T$.
\end{lemma}

\begin{proof}
To prove this lemma one just needs to observe that for any distribution $p$ coming form a Markov process on $T$, $M_i(p)$ is a distribution on $T_i$ that comes from the Markov process that at each edge $e$ has the same transition matrices as $e$ had on the tree $T$.
\end{proof}

Another construction, which is new, and particular to the  $RC$/$EI$ model is described  in the following lemma. This lemma shall be used in section 6 where we provide specific linear invariants for quartet trees.

\begin{lemma}\label{lem_ext}\rm{[Extension lemma]}
Let $\Delta = \sum_{\chi} a_{\chi} \x_{\chi}$ be a linear invariant for an $[n]$--tree $T$ evolving under the $RC$ model.
\begin{enumerate}
\item[(a)] Let $T'$ be the tree obtained by subdividing an edge of $T$ and attaching a new pendant edge at the newly introduced node.
Let $s$ be a new state not involved in $\Delta$ (that is, $a_{\chi}=0$ if $\chi$ contains $s$). Then,
\begin{equation}\label{eqbig}
\sum_{\chi} a_{\chi} \x_{\chi s}
\end{equation}
is a linear invariant for $T'$ (where the new leaf is labelled as leaf $n+1$).
\item[(b)] Let $T'$ be the tree obtained by subdividing an edge of $T$ and attaching a tree $\tilde{T}$ of $m+1$ leaves to the newly introduced node (so that $T'$ has $n+m$ leaves and the newly introduced leaves are labelled from $n+1$ to $n+m$).
Let  $\mu$ be a character on $m$ leaves for which $a_{\chi}=0$ if $\chi$ contains some state in $\mu$
(that is, $\Delta$ does not involve the states of $\mu$ at any leaf). Then
$\sum_{\chi} a_{\chi} \x_{\chi \mu }$ is  a linear phylogenetic invariant for $T'$ (where $\chi \mu$ stands for states $\chi$ at the first $n$ leaves and states $\mu$ at the other $m$ leaves).
\item[(c)] Suppose $T$ is the star tree, and let $\mu$ be a character on $m$ leaves  for which $a_{\chi}=0$ if $\chi$ contains some state in $\mu$.
Then, for the star tree $T'$ with $n+m$ leaves evolving under the $RC$ model, $\sum_{\chi} a_{\chi} \x_{\chi \mu }$ is a linear phylogenetic invariant.

\end{enumerate}
\end{lemma}

\begin{proof} (a) 
By Lemma  \ref{lem_forests}, we only need to check that (\ref{eqbig})  vanishes for the distributions generated with $\Theta=\Theta_F$ where
$F$ is a full subforest of $T'$.
We denote by $\Theta_{F|T}$ the corresponding probabilities at the edges of $T$  and we denote by $\Delta(\Theta_{F|T})$ the value of $\Delta$ evaluated at $\PP_{T,\Theta_{F|T}}$.

If $F$ contains a tree with the new edge $e'$ on it, then, for all $\chi$ involved in $\Delta$, we have $\PP_{T'}(\chi s |\Theta_F )=0$
(because $s$ is a state not involved in $\Delta$) and then \eqref{eqbig} trivially vanishes.
If $F$ does not contain the edge $e'$, then the new leaf is a singleton in $F$. In this case we have $\PP_{T'}(\chi s |\Theta_F)=
\pi_s \PP_{T}(\chi |\Theta_{F|T})$. Therefore \eqref{eqbig} evaluated at $\PP_{T',\Theta_F}$ is $\Delta(\theta_{F|T})$ multiplied by $\pi_s$,
so it vanishes as well.

(b) If $\tilde{T}$ is binary, then the addition of $\tilde{T}$ can be obtained by successively adding cherries to $T$.
So, assume that we have added one cherry as in (a), so that we have assigned state $s$ to the new leaf $l_{n+1}$, and now we add a new cherry
to the edge leading to $l_{n+1}$. Now the new state $s'$ that we consider for the new leaf now can be allowed to be equal to the state $s$ as long as $s'$ differs from the states that appear in $\Delta$.
Indeed, if $s'=s$, there might be forests containing the new cherry, but all of them give probability zero for the states appearing in the
polynomial except if the forest is formed by the new cherry and other trees.  For such a forest $F$ we have $\PP(\chi s s |\Theta_F)=
\pi_s \PP(\chi|\Theta_{F|T})$ and hence the polynomial evaluated at the parameters of this forest is $\Delta(\Theta_{F|T})$ multiplied by $\pi_s$ which vanishes again.

If $\tilde{T}$ is not a binary tree, then it can be also constructed from a binary tree by contracting edges. As for binary tree the polynomial is a phylogenetic invariant,
so it is when we contract edges (note that if a polynomial is a phylogenetic invariant for a tree, then it is also a phylogenetic invariant
for the tree $T_0$ obtained by contracting one edge $e_0$ because any collection of edge parameters at $T_0$ gives a collection of edge parameters
for $T$ by assigning $\theta_{e_0}=0$).

(c) This follows from (b) by contracting edges.
\end{proof}

\section{Phylogenetic mixtures}
\label{Phylogenetic_mixtures}
So far, we have found some linear polynomials that turn out to be either model invariants or topology invariants.
But we were not able to say whether these invariants actually generate the space of linear phylogenetic invariants for a tree $T$.
On the other hand, it would be interesting to know whether a distribution where all these linear invariants vanish
is actually a linear combination from distributions on a tree or a mixture of trees.
To this end, one defines the space of \emph{mixtures} on a tree \citep{stef07}.

\begin{definition} Fix a distribution $\pi$ on the set of states. Given a particular tree $T$, we denote by $\PP_{T,\Theta}$
the distribution of a $RC$ model with parameters $\pi,\Theta$ on $T$. We define the \emph{space of mixtures on $T$} as
$$\mathcal{D}^{\pi}_T=\left\{p=\sum_{i}\lambda_i \PP_{T,\Theta_i} \,\Big| \, \sum_i \lambda_i=1\right\}.$$
If $\mathcal{T}$ is the set of phylogenetic trees on $[n]$, we define \emph{the space of phylogenetic mixtures} on $[n]$ as
$$\mathcal{D}^{\pi}=\left\{p=\sum_{i}\lambda_i \PP_{T_i,\Theta_i} \,\Big| \, \sum_i \lambda_i=1 \, , \, T_i \in \mathcal{T} \right\}$$
\end{definition}

When $\{p_i\}_{i\in I}$ is a set of points in an affine linear space, we denote by $\langle p_i |\, i \in I\rangle_a$
the linear span of these points, that is, the set of points $q=\sum_{i}\lambda_ip_i$ with $\sum_i \lambda_i=1$
(we put the subindex $a$ in order to distinguish this affine linear span from the usual linear span of vectors).
Note that the spaces of phylogenetic mixtures are affine linear varieties,
$$\mathcal{D}^{\pi}_T=\Big\langle p \,\Big|\, p=\PP_{T,\Theta}\Big\rangle_a \, , \quad
\mathcal{D}^{\pi}=\Big\langle p \, \Big|\, p=\PP_{T,\Theta}, \, T \in \mathcal{T} \Big\rangle_a \, ,$$
 and both lie inside the hyperplane
$$H=\left\{\x=(\x_{\chi})_{\chi} \in \RR^{N}\sum_{\chi\in Ch(n,\kappa)} \x_{\chi}=1\right\}.$$

Strictly speaking, for applications in phylogenetics it is only relevant to consider points in $\mathcal{D}^{\pi}$ (or $\mathcal{D}^{\pi}_{T}$) that are actually
distributions. In other words, one should be mainly interested in convex combinations of the points $\PP_{T,\Theta}$:
$$\left\{p=\sum_{i}\lambda_i \PP_{T,\Theta_i} \,\Big| \, \lambda_i \geq 0 , \sum_i \lambda_i=1\right\} \quad \textrm{and}$$
$$\left\{p=\sum_{i}\lambda_i \PP_{T_i,\Theta_i} \,\Big| \, \lambda_i \geq 0 , \sum_i \lambda_i=1 \, , \, T_i \in \mathcal{T} \right\}.$$
However, as the dimension of a polyhedron is the dimension of its affine hull, we focus on computing the dimension of $\mathcal{D}^{\pi}$ and $\mathcal{D}^{\pi}_{T}$.

For any distribution $\pi$, we denote by $L^{\pi}$  the vector space of linear model invariants and by $L^{\pi}_T$
the  space of all linear phylogenetic invariants for a tree $T$.
The orthogonal subspace of $L^{\pi}$ (respectively $L^{\pi}_T$) shall be denoted by $E^{\pi}$ (respectively $E^{\pi}_T$), that is,
$E^{\pi}$ is the set of vectors in $\RR^N$ where all the linear model invariants vanish and $E^{\pi}_T$ the set of vectors where
all the linear phylogenetic invariants for $T$ vanish (by identifying dual and orthogonal spaces).
In other words, $E^{\pi}_T$ and $E^{\pi}$ are spanned by the following vectors of distributions:
$$E^{\pi}_T=\Big\langle  \vec{p}\, \Big|\,  \vec{p}=\PP_{T,\Theta}\Big\rangle \, ,\quad
E^{\pi}_T=\Big\langle  \vec{p}\, \Big|\,  \vec{p}=\PP_{T,\Theta} , \, T \in \mathcal{T}\Big\rangle .$$
Note that when we use $p\in \RR^N$ as a vector, we use the notation $ \vec{p}$ to distinguish it from its use as an affine point in $\RR^N$.
Then the following equalities are clear
$$\mathcal{D}^{\pi}_T =E^{\pi}_T \cap H \, , \quad \mathcal{D}^{\pi}=E^{\pi}\cap H .$$
Therefore, studying phylogenetic mixtures (on $[n]$ or on a tree) is equivalent to studying linear phylogenetic invariants (only model invariants or together with topology invariants). Note that due to Lemma \ref{lem_forests}, it is clear that
$$E^{\pi}_T =\langle  \vec{p}=\PP_{T,\Theta_F} \, | F \in \mathcal{F}_T \rangle \, , \quad E^{\pi} =\langle
 \vec{p}=\PP_{T,\Theta_F} | T \in \mathcal{T} , F \in \mathcal{F}_T \rangle $$
(see also \citet{mat08} Prop. 10).

In this section we compute the dimension of the spaces of phylogenetic mixtures.

\subsection{Model invariants and phylogenetic mixtures}

We fix $n\geq 4$ throughout this section.
We call $\Sigma_{\kappa}$ the set of partitions of $[n]$ of size at most  $\kappa$ (note that if $\kappa \geq n$, this is the whole set of partitions of $[n]$).
{If $\sigma$ is a partition of $[n]$ compatible with trees $T$ and $T'$, and we consider $F=F_{T}(\sigma)$ and $F'=F_{T'}(\sigma)$, then one has $\PP_{T,\Theta_F}=\PP_{T',\Theta_{F'}}.$ This point will be briefly denoted as $q_{\sigma}$ (because it does not depend on the chosen tree compatible with $\sigma$). We give the coordinates of the points $q_{\sigma}$ for $n=4$ shortly, see Example \ref{ex_points}.
Note that $\mathcal{D}^{\pi}=\langle q_{\sigma} \, \mid \, \sigma \in \Sigma_{n} \rangle_a$, but this spanning set of points are not affine linearly independent if $\kappa \geq n$:}

\begin{theorem}\label{thm_model} {If $\pi$ is a distribution on $\kappa$ states with positive entries, then $\left\{ q_{\sigma} \, \mid \, \sigma \in  \Sigma_{\kappa} \right\}$ are affine linearly independent points.
Moreover, if $\pi$ is the uniform distribution or a generic distribution, or if $\kappa \geq n$, then
$\mathcal{D}^{\pi}$ coincides with $\langle q_{\sigma} \, \mid \, \sigma \in \Sigma_{\kappa} \rangle_a$ and has dimension $|\Sigma_{\kappa}|-1$ (which equals $B_{n}-1$ if $\kappa\geq n$).}
\end{theorem}

The inclusion $\langle q_{\sigma} \, \mid \, \sigma \in \Sigma_{\kappa} \rangle_a \subseteq \mathcal{D}^{\pi}$ clearly holds (and if $\kappa\geq n$,
the other inclusion is trivial). The idea for the proof of the other inclusion is to use $\mathcal{D}^{\pi}= E^{\pi}\cap H$,
 bound the dimension of $E^{\pi}$ from above by a quantity $d$ and prove that the set of points $q_{\sigma}$ span an
 affine linear variety of dimension $d-1$. We first need the following lemma.

\begin{lemma}\label{lem_pointsmodel}
\begin{enumerate}
\item[(a)] For any $\kappa$, the set $\left\{ q_{\sigma} \, \mid \, \sigma \in  \Sigma_{\kappa} \right\}$ is formed by affine linearly independent points
for any distribution $\pi$ (with positive entries).
\item[(b)] If $\pi_U$ is the uniform distribution, then the set of linear model invariants is spanned by
the set of polynomials $\x_{\chi}-\x_{\chi'}$ for $\sigma(\chi)=\sigma(\chi')$. In particular, the set of vectors $E^{\pi_U}$
where the model invariants vanish has dimension equal to $|\Sigma_{\kappa}|$.
\end{enumerate}
\end{lemma}

\begin{proof} (a) We need to prove that 
if we have a  linear combination
\begin{equation}\label{eq_li}\sum_{\sigma \in \, \Sigma_{\kappa}}\lambda_{\sigma} q_{\sigma}=0 \end{equation}
with $\sum_{\sigma}\lambda_{\sigma}=0$, then we need to prove that the coefficients $\lambda_{\sigma}$ are zero.
We proceed by induction on $m=\min\{n,\kappa\}$. Note that as all partitions of $[n]$ are of size at most $n$,
$\Sigma_{\kappa}$ equals the set $\Sigma_m$ of partitions of size at most $m$.

If $m=1$, then $\Sigma_{\kappa}$ contains a single element and there is nothing to prove.
Assume that $m\geq 2$ and consider a linear combination as in Eqn.~(\ref{eq_li}).

Note that the coordinate $\tx_{\chi}$ of $q_{\sigma}$ is zero if $\sigma$ does not refine $\sigma(\chi)$.
Let $\tx_{\chi}$ be a coordinate such that $\sigma(\chi)$ has the maximum size $m$.
Then  $\tx_{\chi}$ is different from zero only for $q_{\sigma(\chi)}$
(because the other points $q_{\sigma}$ correspond to partitions that do not refine $\sigma(\chi)$).
Thus, $\lambda_{\sigma(\chi)}=0$ and hence in \eqref{eq_li} we have $\lambda_{\sigma}=0$ for all $\sigma$ of size $m$.
Thus, we are left with a linear combination such as
{$$\sum_{\sigma \in \, \Sigma_{m-1}}\lambda_{\sigma} q_{\sigma}=0 \, , \quad \sum_{\sigma \in \, \Sigma_{m-1}}\lambda_{\sigma} =0.$$}
The result follows by the induction hypothesis.

\noindent
(b) For the uniform distribution, each polynomial $\x_{\chi}-\x_{\chi'}$ for $\sigma(\chi)=\sigma(\chi')$
is clearly a  model invariant (see Remark \ref{rem_symmetric}).
Thus the set of vectors $E^{\pi_U}$ where these polynomials vanish has dimension less than or equal to $|\Sigma_{\kappa}|$.
The set of points considered in (a) for $\pi_U$ is contained in $E^{\pi_U}\cap H$, and hence
(as $H$ is an equation linearly independent with the previous polynomials), the dimension of $E^{\pi_U}$ is $|\Sigma_{\kappa}|$.
It follows that the inclusion $E^{\pi_U}\subseteq \{\x \in \mathbb{R}^{N} | \x_{\chi}=\x_{\chi'} \textrm{ if } \sigma(\chi)=\sigma(\chi')\}$
is actually an equality and the set of model invariants is spanned by the polynomials $\x_{\chi}-\x_{\chi'}$ for $\sigma(\chi)=\sigma(\chi')$.
\end{proof}


%

Now we are ready to prove the theorem.

\begin{proofthmmodel}
We claim that the dimension of $E^{\pi}$ can be bounded from above by the dimension of $E^{\pi_U}$:

\underline{Claim:} For a generic distribution $\pi$, the dimension of $E^{\pi}$ is less than or equal to the dimension $E^{\pi_0}$
for a particular distribution $\pi_0$.

\underline{Proof of Claim:} We think first of the coordinates of $\pi$ as parameters, so that we consider model invariants as linear polynomials
in the variables $\x_{\chi}$ with coefficients in the field of rational functions $\mathbb{R}(\pi)$ (i.e. the field of fractions of the
ring of polynomials $\mathbb{R}[\pi_1,\dots,\pi_{\kappa}]$). The set of all model invariants is a $\mathbb{R}(\pi_1,\dots,\pi_{\kappa})$-vector
space. Consider a basis $l_1,\dots,l_t$ of this space and let $E$ be its orthogonal subspace,
$E=\{\x \in \mathbb{R}^{N}| l_i(\x)=0, i=1,\dots, t\}$ so that $\dim E=N-t$.
When we substitute $\pi$ by a particular value $\pi_0$, $l_1,\dots,l_t$ may not be linearly independent any more,
and the corresponding space $E^{\pi_0}$ may have dimension $\geq \dim E$. But for a generic $\pi$, the dimension of the corresponding space
coincides with dimension of $E$ (because $\pi$ moves in an irreducible space). Therefore,
for a generic $\pi$ we have $\dim E^{\pi}=\dim E \leq \dim E^{\pi_0}$ and the claim is proved. 

By the Claim, for a generic $\pi$, the dimension of $E^{\pi}$ is less than or equal to $\dim E^{\pi_U}$ for the uniform distribution $\pi_U$ and the dimension of this vector space is $|\Sigma_{\kappa}|$ (by Lemma \ref{lem_pointsmodel}(b)). Thus, $\dim E^{\pi} \leq |\Sigma_{\kappa}|$. On the other hand, the dimension of $\langle q_{\sigma} \, \mid \, \sigma \in \Sigma_{\kappa}$ is $|\Sigma_{\kappa}|-1$ by Lemma \ref{lem_pointsmodel}(a). The inclusion
 $$\langle q_{\sigma} \, \mid \, \sigma \in \Sigma_{\kappa} \rangle \subseteq \mathcal{D}^{\pi} = E^{\pi}\cap H$$
finishes the proof. Note that if $\kappa \geq n$ one immediately has $\mathcal{D}^{\pi}=\langle q_{\sigma} \, \mid \, \sigma \in \Sigma_{n}$ for any $\pi$, and its dimension follows from Lemma \ref{lem_pointsmodel}(a).
\end{proofthmmodel}


\begin{remark}\rm{
In Theorem \ref{thm_model} we give a  set of affine independent points that span $\mathcal{D}^{\pi}$ for almost any distribution $\pi$. From this set of points (vectors) it easy to compute a basis of the space of linear invariants $L^{\pi}$ as its orthogonal space.
}
\end{remark}

\begin{example}\label{ex_points}\rm
We give here the coordinates of the points that span the spaces of mixtures on trees with $n=4$ and $\kappa=4$ or $\kappa=3$.

For $\kappa=4$ we have $|\Sigma_4|=B_4=15$ and $\mathcal{D}^{\pi}=\langle q_{\sigma} \, \mid \, \sigma \in \Sigma_{\kappa} \rangle$.
We start with 12 partitions $\sigma$ that correspond to forests in the star tree $T_{*}$.
We call $q_{\bullet}$ the point corresponding to the  trivial subforest of $T_*$ (formed by singletons).
We call $q_{ij}$ the points corresponding to the full subforest of $T_*$ formed by the tree $T[i,j]$ and singletons
(this gives six points, $q_{ij}$, $i< j$).
Then we consider the forests formed by  a subtree of three leaves $i,j,k$ and a singleton, which gives four points
$q_{123}$, $q_{124}$, $q_{134}$, $q_{234}$. Finally, we denote by  $q_{1234}$ the point corresponding to the forest $F=\{T_*\}$.
To simplify notation we write the normalized coordinates $\tx_{\chi_1\dots\chi_4}$ instead of $\x_{\chi_1\dots\chi_4}$.
Let the space of states $S$ be $\{x,y,z,w\}$.
In order to prove that the 15 points we provide are affine linearly independent, it is enough to look at the following
15 coordinates of these points: $$\tx_{xxxx}, \tx_{xxxy},\tx_{xxyx},\tx_{xyxx},\tx_{yxxx},\tx_{ xxyy},\tx_{xyxy},\tx_{xyyx},$$
$$\tx_{ xxyz},\tx_{xyxz},\tx_{xyzx},\tx_{yxzx},\tx_{yxxz},\tx_{yzxx}, \tx_{xyzw}.$$

In Table \ref{table_star}
we write the coordinates of the first 12 points considered above.


\begin{table}
\caption{\label{table_star}Linearly independent points for $\mathcal{D}_{T_*}$ for $n=4$ in coordinates $\tx's$}
\resizebox{12cm}{!}{
\begin{tabular}{l||c|c|c|c|c|c|c|c|c|c|c|c|c|c|c|}

               & xxxx & xxxy & xxyx & xyxx & yxxx& xxyy & xyxy & xyyx & xxyz & xyxz & xyzx & yxxz & yxzx & yzxx & xyzw\\

                \hline
               $q_{\bullet}$& 1 & 1 & 1 & 1 & 1 & 1 & 1 & 1 & 1 & 1 & 1 & 1 & 1 & 1 & 1 \\
               $q_{12}$&$\frac{1}{\pi_x}$ & $\frac{1}{\pi_x}$ & $\frac{1}{\pi_x}$ & 0 &0 & $\frac{1}{\pi_x}$ & 0 & 0 &$\frac{1}{\pi_x}$&0 & 0 & 0 & 0 & 0 & 0 \\
               $q_{13}$&$\frac{1}{\pi_x}$ & $\frac{1}{\pi_x}$ & 0 & $\frac{1}{\pi_x}$ &0 & 0 & $\frac{1}{\pi_x}$ & 0 & 0 & $\frac{1}{\pi_x}$  & 0 & 0 & 0 & 0 & 0 \\
               $q_{14}$&$\frac{1}{\pi_x}$ & 0& $\frac{1}{\pi_x}$ & $\frac{1}{\pi_x}$ &0 & 0 & 0 & $\frac{1}{\pi_x}$ & 0 & 0 & $\frac{1}{\pi_x}$  & 0 & 0 & 0 & 0 \\
               $q_{23}$&$\frac{1}{\pi_x}$ & $\frac{1}{\pi_x}$ & 0 & 0& $\frac{1}{\pi_x}$ &0 & 0 & $\frac{1}{\pi_y}$ & 0 & 0  & 0 & $\frac{1}{\pi_x}$  & 0 & 0 & 0 \\
               $q_{24}$&$\frac{1}{\pi_x}$ & 0 & $\frac{1}{\pi_x}$ & 0 & $\frac{1}{\pi_x}$ & 0 & $\frac{1}{\pi_y}$ & 0 & 0 & 0 & 0 &0 & $\frac{1}{\pi_x}$  & 0 & 0 \\
               $q_{34}$&$\frac{1}{\pi_x}$ & 0 & 0 & $\frac{1}{\pi_x}$ & $\frac{1}{\pi_x}$ & $\frac{1}{\pi_y}$ & 0 & 0 & 0 & 0 & 0 & 0 & 0 & $\frac{1}{\pi_x}$  & 0 \\
               $q_{123}$&$\frac{1}{\pi_x^2}$ & $\frac{1}{\pi_x^2}$ & 0 & 0 & 0 & 0 & 0 & 0 & 0 & 0 & 0 & 0 & 0 & 0 & 0 \\
               $q_{124}$&$\frac{1}{\pi_x^2}$ & 0 & $\frac{1}{\pi_x^2}$ & 0  & 0 & 0 & 0 & 0 & 0 & 0 & 0 & 0 & 0 & 0 & 0 \\
               $q_{134}$&$\frac{1}{\pi_x^2}$ & 0 & 0 & $\frac{1}{\pi_x^2}$ & 0 & 0 & 0 & 0 & 0 & 0 & 0 & 0 & 0 & 0 & 0 \\
               $q_{234}$&$\frac{1}{\pi_x^2}$ & 0 & 0 & 0 & $\frac{1}{\pi_x^2}$ & 0 & 0 & 0 & 0 & 0 & 0 & 0 & 0 & 0 & 0 \\
               $q_{1234}$&$\frac{1}{\pi_x^3}$ & 0 & 0 & 0 & 0 & 0 & 0 & 0 & 0 & 0 & 0 & 0 & 0 & 0 & 0 \\
               \hline
             \end{tabular}
						 }
\end{table}
If we consider the previous points plus the point $q_{12|34}$  that corresponds to the forest $\{T[1,2], T[3,4]\}$ on the tree $T_{12|34}$, then we obtain a set of linearly independent points that span $\mathcal{D}^{\pi}_{12|34}$. In Table \ref{table_1234} we show the coordinates of this new point.
\begin{table}
\caption{\label{table_1234} The new point added for tree $12|34$}
\resizebox{12cm}{!}{
\begin{tabular}{l||c|c|c|c|c|c|c|c|c|c|c|c|c|c|c|}
& xxxx & xxxy & xxyx & xyxx & yxxx& xxyy & xyxy & xyyx & xxyz & xyxz & xyzx & yxxz & yxzx & yzxx & xyzw\\
\hline
$q_{12|34}$& $\frac{1}{\pi_x^2}$& 0 & 0& 0& 0&$\frac{1}{\pi_x\pi_y}$ &0&0&0&0&0&0&0&0&0
\end{tabular}
}
\end{table}

Now we consider the points corresponding to the forests compatible  for the remaining quartets, $q_{13|24}$, $q_{14|23}$
(their coordinates are shown in Table \ref{table_other}). The previous points together with these two points span the space of mixtures $\mathcal{D}^{\pi}$.

\begin{table}
\caption{\label{table_other} The two points added when considering the quartets $13|24$ and $14|23$}
\resizebox{12cm}{!}{
\begin{tabular}{l||c|c|c|c|c|c|c|c|c|c|c|c|c|c|c|}
& xxxx & xxxy & xxyx & xyxx & yxxx& xxyy & xyxy & xyyx & xxyz & xyxz & xyzx & yxxz & yxzx & yzxx & xyzw\\
\hline
$q_{13|24}$& $\frac{1}{\pi_x^2}$& 0 & 0& 0& 0& 0&$\frac{1}{\pi_x\pi_y}$ &0&0&0&0&0&0&0&0\\
$q_{14|23}$& $\frac{1}{\pi_x^2}$& 0 & 0& 0& 0&  0 &  0&$\frac{1}{\pi_x\pi_y}$ &0&0&0&0&0&0&0
\end{tabular}
}
\end{table}

\noindent

Consider now the case $\kappa=3$. Then, according to Theorem \ref{thm_model}, $\mathcal{D}^{\pi}$ has dimension 13 for generic $\pi$.
Indeed, if we consider the  15 points above, then they are no longer linearly independent
when the last column of the table is removed. The last 14 points suffice to span $\mathcal{D}^{\pi}$ in this case.

\end{example}

\section{Phylogenetic mixtures on a  fixed tree}
\label{Phylogenetic_mixtures_on_a_fixed_tree}





In this section we compute the dimension of the space of phylogenetic mixtures on a tree, give an algorithm to compute a basis of the space of liner topology invariants and
we explain whether Lake--type invariants of Proposition \ref{lake} suffice to describe the space of phylogenetic invariants.
For $\kappa=2$ there are known to be no linear topology invariants \citep{mat08};  these arise for $\kappa \geq 3$ (see Lemma \ref{lem_3states} below, though  Lake--type invariants only appear when $\kappa \geq 4$).
Moreover, even when $\kappa=4$ for certain models there exist other linear topology invariants beyond the Lake--type ones \citep{fu95}.
By considering the $EI/RC$ model we show how it is possible to characterize the quotient space of linear topology invariants
for any number of states and taxa,
and provide an explicit algorithm for constructing a basis for  the (quotient) space of topological invariants.
As explained in the introduction, linear topology invariants are of interest because they provide a way to distinguish distributions coming from mixtures on
a particular topology from distributions arising as mixtures on another topology.

Recall that $E_T^{\pi}$ is the space of vectors where the linear phylogenetic invariants vanish.
We know by Lemma \ref{lem_forests}(b) that a homogeneous linear polynomial vanishes on
all distributions $\PP_{T,\Theta}$ if and only if it vanishes on all distributions of type $\PP_{T,\Theta_F}$ for $F$ a full subforest of $T$.
Therefore we have
 $$E_T^{\pi}=\langle  { \vec{q}}_F | \,  \, F \in \mathcal{F}_T \rangle.$$

\begin{example}\rm\label{casen=3}
 Let $n=3$, let $T$ be the tripod tree and assume that $\kappa\geq 3$. We prove here that the vectors ${ \vec{q}}_F$, for $F \in \mathcal{F}_T$ are linearly independent.
 These vectors are: ${ \vec{q}}_{\bullet}$ corresponding to the trivial subforest, ${ \vec{q}}_{12|3}$, ${ \vec{q}}_{13|2}$, ${ \vec{q}}_{23|1}$ corresponding
 to full sub forests with one singleton, and  ${ \vec{q}}_{123}$ corresponding to the  tree itself.
 We choose three states $x,y,z$ and we provide in Table \ref{table_casen=3} the submatrix corresponding to the coordinates $\x_{xxx}$, $\x_{xxy}$, $\x_{xxy}$,
 $\x_{xyx}$, $\x_{yxx}$, $\x_{xyz}$. It is clear that this submatrix has nonvanishing determinant if $\pi$ is positive.

\begin{table}
\caption{\label{table_casen=3} Table of example \ref{casen=3}.}
\begin{tabular}{l||c|c|c|c|c|}
$\x_{xxx}$ & $\x_{xxy}$ & $\x_{xxy}$ &
 $\x_{xyx}$ & $\x_{yxx}$ & $\x_{xyz}$\\
\hline
${ \vec{q}}_{\bullet}$ &$\pi_x^3$ &$\pi_x^2\pi_y$& $\pi_x^2\pi_y$ & $\pi_x^2\pi_y$ & $\pi_x\pi_y\pi_z$\\
${ \vec{q}}_{12|3}$ &$\pi_x^2$ & $\pi_x\pi_y$ & 0&0&0\\
${ \vec{q}}_{13|2}$ &$\pi_x^2$& 0&$\pi_x\pi_y$ &0&0\\
${ \vec{q}}_{23|1}$ &$\pi_x^2$ &0 &0 &$\pi_x\pi_y$ &0\\
${ \vec{q}}_{123}$ & $\pi_x$ & 0 &0&0&0
\end{tabular}
\end{table}

\end{example}



Let $T$  be a binary tree on $[n]$, $n\geq 4$, and assume that leaves $n$ and $n-1$ form a cherry $c$.
Let $u$ be the interior node of this cherry, and let  $e$ be the edge adjacent to $u$ and not to $n,n-1$.
Let $T'$ be the subtree $T-\{e_n,e_{n-1}\}$.
We denote by $\mathcal{F}_c$ the set of full subforests of $T$ that contain a tree with the cherry $c=\{e_n,e_{n-1}\}$.
For any leaf $l$ we let $\mathcal{F}_l$ be the set of full subforests of $T$ that contain $l$ as a singleton
and we call $T_l$ the tree obtained by replacing the two edges adjacent to $e_l$ by a single edge. 
Then $\mathcal{F}_{T}$ is the disjoint union of $\mathcal{F}_{c}$ and $\mathcal{F}_{n-1}\cup \mathcal{F}_{n}$.

\begin{lemma}\label{lem_iso} For a binary tree on $n \geq 4$ leaves we have isomorphisms of vector spaces:
$$\langle  { \vec{q}}_F | \,  \, F \in \mathcal{F}_l \rangle \cong \langle  { \vec{q}}_G | \,  \, G \in \mathcal{F}_{T_l} \rangle \, , \quad \langle  { \vec{q}}_F | \,  \, F \in \mathcal{F}_c \rangle \cong \langle  { \vec{q}}_G | \,  \, G \in \mathcal{F}_{T'} \rangle .$$
\end{lemma}

\begin{proof} We start with the first isomorphism. For simplicity we assume $l=n$ (and for this isomorphism $n$ is not necessarily a leaf in a cherry). Let $V_n$ be the vector space $\langle  { \vec{q}}_F | \,  \, F \in \mathcal{F}_n \rangle$.
For any state $s \in S$ we denote by $f^{s}$ the projection map from  $\mathbb{R}^{\kappa^n}$
to the subspace $R_s$ corresponding to coordinates coordinates $\x_{\chi_{1} \dots\chi_{n-1} s}$, so that we can view $\mathbb{R}^{\kappa^n}$ as the direct sum $R_{s_1}\oplus \dots\oplus R_{s_{\kappa}}.$
For a vector $v \in \mathbb{R}^{\kappa^n}$ we denote by $(f^{s_1}(v),\dots,f^{s_{\kappa}}(v))$ the decomposition of $v$ according to this direct sum.
Note that if $F\in \mathcal{F}_n$, then $\PP_T(\chi_1\dots\chi_n|\Theta_F)=\pi_{\chi_n}\PP_{T'}(\chi_1\dots\chi_{n-1}|\Theta_{F|T_n})$.
In particular, we have $f^{s}({ \vec{q}}_F)=\pi_{s}{ \vec{q}}_{F|T'}$ for any $s\in S$ and ${ \vec{q}}_{F}=(\pi_{s_1}{ \vec{q}}_{F|T_n},\dots,\pi_{s_{\kappa}}{ \vec{q}}_{F|T_n})$.

We prove here that (for any $s\in S$) the linear map $f^s$ is an isomorphism between $V_n$ and the target vector space.
First of all, the linear map $f^s_{|V_n}$ is injective.
Indeed, if $f^s_{|V_n}(v)=0$ for a certain $v=\sum_{F \in \mathcal{F}_n} \lambda_F{ \vec{q}}_F$, then $0=\sum_{F \in \mathcal{F}_n} \lambda_F f^s({ \vec{q}}_F)=\sum_{F \in \mathcal{F}_n} \lambda_F \pi_s { \vec{q}}_{F|T_n}$
and hence (assuming $\pi_s\neq0$) $\sum_{F \in \mathcal{F}_n} \lambda_F  { \vec{q}}_{F|T'}=0$.
This implies that $v=(0,\dots,0)$ in $R_{s_1}\oplus \dots\oplus R_{s_{\kappa}}$ and so $f^s_{|V_n}$ is an injective linear map.


We prove that the image of $f^s_{|V_n}$ is $\langle  { \vec{q}}_G | \,  \, G \in \mathcal{F}_{T_n} \rangle$. From the above, one can easily see that ${\rm Im} f^s_{|V_n}$ is contained in $\langle  { \vec{q}}_G | \,  \, G \in \mathcal{F}_{T_n} \rangle$. Now for any $G \in \mathcal{F}_{T_n} $ we shall find $\tilde{G}\in \mathcal{F}_{T} $ such that  $\tilde{G}_{|T_n}=G.$
If $n$ does not belong to a cherry, we consider $\tilde{G}$ to be the full subforest of $T$ defined by the singleton $\{n\}$, and the trees in $G$ (thinking of $T_n$ as a subtree of $T$).
If $n$ belongs to a cherry, we can think of $T_n$ as the tree $T'$ described above. Now for any  $G\in  \mathcal{F}_{T'} $, we consider $\tilde{G}$ the full subforest of $T$ defined by: the singleton $\{n\}$, $t$ for any  $t\in G$ not containing $e$ nor $u$,   $t\cup {e_{n-1}}$ if there is $t\in G$ containing $e$, and the singleton $\{n-1\}$ if
$G$ contains the singleton $\{u\}$.  In this way we have $\tilde{G}_{|T'}=G$ and ${ \vec{q}}_{G}=\frac{1}{\pi_s}f^s_{|V_n}{ \vec{q}}_{\tilde{G}}\in {\rm Im} f^s_{|V_n}$, so the other inclusion is proved.

As far as the second isomorphism is concerned, we consider the subspace $L\subset \mathbb{R}^{\kappa^n}$ given by coordinates of type $\x_{\chi_1\dots\chi_{n-2}s s}$ for any $\chi_1,\dots,\chi_{n-2},s$ in $S$.
We have $\mathbb{R}^{\kappa^n}=L\oplus L^{\perp}$ and if $f$ denotes the projection to $L$, then any vector $v$ can be decomposed as
$(f(v),v-f(v)).$ If $F\in \mathcal{F}_c$, then $\PP_T({\chi_1\dots\chi_{n-1}\chi_{n}}|\Theta_F)$
is zero if $\chi_{n-1}\neq \chi_n$ and is equal to $\PP_{T'}(\chi_1\dots\chi_{n-2}|\Theta_{F|T'})$ if $\chi_{n-1}=\chi_n=s$.
Hence, if $F\in \mathcal{F}_c$ we have ${ \vec{q}}_F=(f({ \vec{q}}),0)=({ \vec{q}}_{F|T'},0).$
Now we prove that $f_{|V_c}$ is injective. Let $v=\sum_{F \in \mathcal{F}_c} \lambda_F{ \vec{q}}_F$ and suppose that $f(v)=0$. Then $0~=~\sum_{F \in \mathcal{F}_c} \lambda_F f({ \vec{q}}_F)=\sum_{F \in \mathcal{F}_c} \lambda_F { \vec{q}}_{F|T'}$ and
$$v=\sum_{F \in \mathcal{F}_c} \lambda_F { \vec{q}}_F=\sum_{F \in \mathcal{F}_c} \lambda_F ({ \vec{q}}_{F|T'},0)=(\sum_{F \in \mathcal{F}_c} \lambda_F { \vec{q}}_{F|T'},0)=0. $$
This proves that $f_{|V_c}$ is injective. Moreover the image of this map is included in the subspace
$\langle  { \vec{q}}_G | \,  \, G \in \mathcal{F}_{T'} \rangle$. For any $G \in \mathcal{F}_{T'}$ we consider the full subforest $\bar{G}$ of $T$ defined by:
the trees in $G$ that do not contain $e$, $t\cup{c}$ if $t$ contains $e$, and the cherry $c$ if $G$ contains the singleton $\{u\}$.
Therefore we have $\bar{G}_{|T'}=G$ and ${ \vec{q}}_{G}=f_{|V_c}{ \vec{q}}_{\tilde{G}}\in {\rm Im} f^s_{|V_c}.$
\end{proof}

\begin{theorem}\label{thm_dimension}
Let $T$ a phylogenetic tree on $n$ leaves, $n\geq 3$, evolving under the $EI/RC$ model for any distribution $\pi$ on $\kappa \geq 3$ states. Then, $\left\{ q_F | \,  \, F \in \mathcal{F}_T \right\}$ are affine independent points that span the space of phylogenetic mixtures on $T$, $\mathcal{D}^{\pi}_T$.
In particular, the dimension of $\mathcal{D}^{\pi}_T$ is $|\mathcal{F}_T|-1$ and when $T$ is binary this dimension is equal to the Fibonacci number $F_{2n-1}$ minus 1.
\end{theorem}

\begin{proof} We proceed by induction on $n$. The statement of the theorem is equivalent to $\dim E^{\pi}_T=|\mathcal{F}_T|$.

The cases $n=3$ and $n=4$ are handled by Examples \ref{ex_points} and \ref{casen=3}.

For  $n\geq 5$,  suppose first that $T$ is a binary tree. We may assume that the statement is true for trees with strictly less than $n$ leaves.
We suppose that $n$ and $n-1$ form a cherry and adopt the notation fixed above. Then we have that
$$E^{\pi}_T=\langle  { \vec{q}}_F | \,  \, F \in \mathcal{F}_T \rangle=\langle  { \vec{q}}_F | \,  \, F \in \mathcal{F}_{n-1} \cup  \mathcal{F}_{n}
\rangle+\langle  { \vec{q}}_F | \,  \, F \in \mathcal{F}_c \rangle.$$
Note that $\langle  { \vec{q}}_F | \,  \, F \in \mathcal{F}_{n-1} \cup  \mathcal{F}_{n}
\rangle$ equals $\langle  { \vec{q}}_F | \,  \, F \in \mathcal{F}_{n-1}\rangle +\langle  { \vec{q}}_F | \,  \, F \in \mathcal{F}_{n}
\rangle$. We know that $\langle  { \vec{q}}_F | \,  \, F \in \mathcal{F}_{n-1}\rangle$ and $\langle  { \vec{q}}_F | \,  \, F \in \mathcal{F}_{n}
\rangle$ have dimension  $|\mathcal{F}_{T'}|$ by Lemma \ref{lem_iso} and the induction hypothesis.
These subspaces intersect in $\langle  { \vec{q}}_F | \,  \, F \in \mathcal{F}_{n-1}\cap \mathcal{F}_{n} \rangle.$
By Lemma \ref{lem_iso} (applied twice) and the induction hypothesis, this linear space has dimension $|\mathcal{F}_{T''}|$ where $T''$ is a tree on $n-2$ leaves.
Therefore, using Grassmann's formula  ($\dim(U + W) = \dim U + \dim W - \dim(U \cap W)$ for subspaces $U, W$ of a vector space)  we have that $\dim (\langle  { \vec{q}}_F | \,  \, F \in \mathcal{F}_{n-1}\rangle +\langle  { \vec{q}}_F | \,  \, F \in \mathcal{F}_{n}
\rangle)=|\mathcal{F}_{T'}|+|\mathcal{F}_{T'}|-|\mathcal{F}_{T''}|$.
As all of these trees are binary, this dimension equals the Fibonacci number $F_{2n-2}$ since $F_{2n-2}=F_{2n-3}+F_{2n-3}-F_{2n-5}$.

On the other hand, by Lemma \ref{lem_iso} and the induction hypothesis, $\langle  { \vec{q}}_F | \,  \, F \in \mathcal{F}_c \rangle$ has dimension $|\mathcal{F}_{T'}|=F_{2n-3}$.
Let us prove now that $\langle  { \vec{q}}_F | \,  \, F \in \mathcal{F}_c \rangle$ and $\langle  { \vec{q}}_F | \,  \, F \in \mathcal{F}_{n-1} \cup  \mathcal{F}_{n}
\rangle$ only intersect in the zero vector.
Let $v$ be a vector in the intersection,
$$v=\sum_{F\in \mathcal{F}_{n-1} \cup  \mathcal{F}_{n} }\lambda_F { \vec{q}}_F =
\sum_{G\in F_c}\mu_G { \vec{q}}_G.$$
Looking at the right-hand side we see that all the coordinates of $v$ of type
$\x_{\chi_{1} \dots\chi_{n-2} s s'}$ for $s\neq s'$ are zero. Let us fix $\chi_{1}, \dots,\chi_{n-2},s\in S$ and we shall prove that
the coordinate $\x_{\chi_{1} \dots\chi_{n-2} s s}$ of $v$, $\x_{\chi_{1} \dots\chi_{n-2} s s}(v)$, is 0.
Let us split the sum $\sum_{F\in \mathcal{F}_{n-1} \cup  \mathcal{F}_{n}}$ into two terms (although this decomposition may not be unique):
$\sum_{F\in \mathcal{F}_{n-1}}\lambda_F { \vec{q}}_{F}+
\sum_{H\in \mathcal{F}_{n} }\lambda_H { \vec{q}}_{H}$.
We denote by $F'$ the restriction of a forest $F$ to $T'$. Note that
$$\x_{\chi_{1} \dots\chi_{n-2} s s}(v)=\pi_s\x_{\chi_{1} \dots\chi_{n-2} s}\left(\sum_{F\in \mathcal{F}_{n-1}}\lambda_F { \vec{q}}_{F'}\right)+
\pi_s\x_{\chi_{1} \dots\chi_{n-2} s}\left(\sum_{H\in \mathcal{F}_{n} }\lambda_H { \vec{q}}_{H'}\right).$$

For each $\alpha \in S$ we denote by $a(\alpha)$
the value of the coordinate $\x_{\chi_{1} \dots\chi_{n-2} \alpha}$ of $\sum_{F\in \mathcal{F}_{n-1}}\lambda_F { \vec{q}}_{F'}$ and by $b(\alpha)$ the
value of this coordinate at $\sum_{H\in \mathcal{F}_{n} }\lambda_H { \vec{q}}_{H'}$. We want to prove that $a(s)+b(s)=0$.
Consider $s'$ and $s''$ states in $S$ different from  $s$ (this is possible because $\kappa\geq 3$). As
$$0=\x_{\chi_{1} \dots\chi_{n-2} s s'}(v)=\pi_{s'}a(s)+\pi_s b(s'),$$
$$0=\x_{\chi_{1} \dots\chi_{n-2} s' s}(v)=\pi_{s}a(s')+\pi_{s'} b(s),$$
$$0=\x_{\chi_{1} \dots\chi_{n-2} s' s''}(v)=\pi_{s''}a(s')+\pi_{s'} b(s''), \mbox{ and } $$
$$0=\x_{\chi_{1} \dots\chi_{n-2} s'' s'}(v)=\pi_{s'}a(s'')+\pi_{s''} b(s'),$$ we have
$$a(s)+b(s)=-\frac{\pi_s}{\pi_{s'}}(b(s')+a(s'))=\frac{\pi_s'}{\pi_{s''}}\frac{\pi_s}{\pi_{s'}}(a(s'')+b(s'')).$$
But now we use the analogous relations between $a(s), a(s'')$, $b(s)$, $b(s'')$:
$$0=\x_{\chi_{1} \dots\chi_{n-2} s s''}(v)=\pi_{s''}a(s)+\pi_s b(s'') \mbox{ and } $$
$$0=\x_{\chi_{1} \dots\chi_{n-2} s'' s}(v)=\pi_{s}a(s'')+\pi_{s''} b(s),$$ in order to obtain that $a(s)+b(s)=-\frac{\pi_s}{\pi_{s''}}(b(s'')+a(s''))$.
Therefore, $a(s)+b(s)=-a(s)-b(s)$ and this quantity vanishes.

Applying Grassmann's formula again, we have $\langle  { \vec{q}}_F | \,  \, F \in \mathcal{F}_{n-1}\cup \mathcal{F}_{n}\rangle \cap  \langle  { \vec{q}}_F | \,  \, F \in \mathcal{F}_c \rangle= 0$ and
$$\dim E^{\pi}_T=\dim(\langle  { \vec{q}}_F | \,  \, F \in \mathcal{F}_{n-1}\rangle +\langle  { \vec{q}}_F | \,  \, F \in \mathcal{F}_{n}
\rangle)+\dim \langle  { \vec{q}}_F | \,  \, F \in \mathcal{F}_c \rangle.$$
We have already seen that the first term is equal to $F_{2n-2}$. The second term is equal to $F_{2n-3}$ by Lemma \ref{lem_iso} and the induction hypothesis.
Therefore $\dim E^{\pi}_T=F_{2n-1}=|\mathcal{F}_T|.$

Let us assume now that $T$ is not binary. We already know that $E^{\pi}_T=\langle  { \vec{q}}_F | \,  \, F \in \mathcal{F}_T \rangle$ and we only need
to check that the vectors ${ \vec{q}}_F $, $F \in \mathcal{F}_T$, are linearly independent. As the forests in $T$ are also subforests of any
binary tree that refines $T$, these vectors are linearly independent by the binary tree case proved above. This finishes the proof.

%
%
\end{proof}

Recall that $L^{\pi}=(E^{\pi})^{\perp}$ and $L^{\pi}_T=(E^{\pi}_T)^{\perp}$ and therefore the quotient space $L^{\pi}_T / L^{\pi}$ of linear
\emph{topology} invariants
is isomorphic to $E^{\pi}/E^{\pi}_T$. As an immediate consequence of Theorems \ref{thm_model} and \ref{thm_dimension} we have:
\begin{corollary}\label{cor_topoinvar}
The dimension of the space of linear topology invariants is $|\Sigma_k|-|co(T)|$ if $\pi$ is either a generic distribution
or the uniform distribution, or $\kappa\geq n$ (and in this last case the dimension equals $|Inc(T)|$).
\end{corollary}

As a consequence of Theorem \ref{thm_dimension}, we are able to provide an algorithm to obtain a basis of the space of linear {topology}
invariants for any tree $T$, $L^{\pi}_T / L^{\pi}$. To do so, note that if  $proj$ is the orthogonal projection from $E^{\pi}$
to the subspace $L^{\pi}_T=(E^{\pi}_T)^{\perp}$,
then $proj$ provides an isomorphism between $E^{\pi}/E^{\pi}_T$ and $L^{\pi}_T / L^{\pi}$ and therefore we have:

\vspace*{1mm}

\noindent\textbf{Algorithm.}
\begin{enumerate}
\item For each $F \in \mathcal{F}_T$ compute the coordinates of the vector $  \vec{q}_F  \in E^{\pi}_T$.
\item Complete the basis $\left\{  \vec{q}_F | \,  \, F \in \mathcal{F}_T \right\}$ by vectors $v_1,\dots,v_d$ from $E^{\pi}$
in order to obtain a basis of $E^{\pi}.$
\item Then the classes of $proj(v_1),\dots,$ $proj(v_d)$ form a basis of the space of linear topology invariants $L^{\pi}_T / L^{\pi}$.
\end{enumerate}

\noindent
Note that step $2$ can be done using the Steinitz exchange lemma and the spanning set of vectors of $E^{\pi}$ provided in
Theorem \ref{thm_model}.
\vspace*{1mm}

We prove now that Lake--type invariants  suffice to define the space of linear topology invariants of a tree when $\kappa\geq n$ and $\pi$ is the uniform distribution. We first need a combinatorial lemma.
\begin{lemma}\label{combinatorial} For any phylogenetic tree $T$ on $[n]$ and any partition $\sigma$ that is incompatible with $T$ there exist
two blocks $B, B'$ of $\sigma$ and leaves $x \in B$, $x' \in B'$ and an interior vertex $v$ of $T$ in the path connecting $x$ and $x'$ for which
the following holds:
\begin{itemize}
\item[]
For each leaf $l$ of $T$  in the same connected component of $T-v$ as $x$,  $l \in B$ or $\{l\} \in \sigma$.
\item[]
For each leaf $l$ of $T$  in the same connected component of $T-v$ as $x'$,  $l \in B'$ or $\{l\} \in \sigma$.
\end{itemize}
\end{lemma}

\begin{proof}
First suppose that $\sigma$ has no singleton blocks. Let us say that an edge  $e= \{u,v\}$ of $T$ is {\em terminating} if:
\begin{itemize}
\item[(i)] all the leaves of $T$ that are in the  subtree $t_e$ of $T-v$ containing $u$  are contained in a single block of
$\sigma$ (say, $B_i$), and
\item[(ii)] at least two of the other subtrees of $T-v$ contain elements of $[n]$ not in $B_i$.
\end{itemize}
For each such terminating edge $e$ delete the pendant subtree $t_e$ from $T$ and label $u$ by $B_i$.
Let $T'$ be the resulting tree.  This tree $T'$ has at least four leaves
(since $\sigma$ is incompatible with $T$) and so $T'$ has a cherry (two leaves that are adjacent to a shared
vertex $v$).  This vertex $v$ and the label sets of the incident leaves ($B$ and $B'$) then satisfies the property claimed in the lemma.
The extension to allow $\sigma$ to have singleton blocks is now straightforward -- we can simply delete them first, repeat the argument above, and add them in afterwards.
\end{proof}

\begin{corollary}\label{cor_Lake} If $\pi_U$ is the uniform distribution and $\kappa\geq n$, then the Lake--type invariants of Proposition
\ref{lake} and model invariants generate the space of linear phylogenetic invariants for $T$.
\end{corollary}

\begin{proof} We omit the superscript $\pi_U$ for the spaces of linear invariants in this proof.
By Lemma \ref{lem_pointsmodel}(b) the  space of model invariants $L$ is spanned by the polynomials $\x_{\chi}-\x_{\chi'}$ for
$\sigma(\chi)=\sigma(\chi')$ and has dimension $\kappa^n-|\Sigma_n|$ (because $\kappa\geq n$).
We also have that $\dim L_T=\kappa^n- \dim E^{\pi_U}_T=\kappa^n-|\mathcal{F}_T|=\kappa^n -(|\Sigma_n|-|{\rm Inc}(T)|)$ and
$\dim L =\kappa^n-\dim E^{\pi_U}=\kappa^n-|\Sigma_n|$. Hence, we have $\dim L_T/L= \dim L_T -\dim L= |{\rm Inc}(T)|$.
So we need to prove that Lake's invariants give a set of $|{\rm Inc}(T)|$ linearly independent vectors in $L_T/L$.

Note that in $L_T/L$ we can work with polynomials in indeterminates $\x_{\sigma}$, $\sigma \in \Sigma_n$.

Let us prove that, if $\sigma$ is an incompatible partition on $T$, then $\x_{\sigma}$ is a linear combination of $\x_{\sigma'}$
for compatible partitions $\sigma'$ of size $>|\sigma|$. To this end, we proceed by induction on $m=n-|\sigma|$.

If $m=0$ or $1$, then $\sigma$ is convex on $T$ and there is nothing to prove.
Let $m\geq 2$ and assume that we have proved the statement when $n-|\sigma|$ is smaller than $m$.
Let $\sigma=\{B_1,\dots,B_r\}$ and we call $s_1,\dots,s_r$ the states associated to $\sigma$.
Assume first that  $\sigma$ has no singletons. Then, according to Lemma \ref{combinatorial} we can find two blocks of $\sigma$, say $B_1$, $B_2$, and an interior vertex $v$ for which all leaves in one of the subtrees $T'_1$ of $T-v$ are in $B_1$, and all leaves in one of the other subtrees $T'_2$ of $T-v$ are in $B_2$. We write $l_i'$ for the set of leaves in $T'_i$ so that $B_i$ is the disjoint union of $l_i'$ and another set $l_i$. We let $\mathcal{E}$ be the event that leaves $B_i$ are in state $s_i$ for $i\geq 3$, leaves in $l_1$ are in state $s_1$ and leaves in $l_2$ are in state $s_2$.
As the fully symmetric model satisfies the partial separability property (PS) and as $|\sigma|\leq n-2 \leq \kappa-2$, we can consider two new
states $s_1'$, $s_2'$ to apply Proposition \ref{lake} (with $t=T'_1$ and $t'=T'_2$). Thus we obtain the following linear invariant (written in
terms of partitions because the states do not matter, as soon as they are different):
$$\x_{\sigma}+\x_{l_1|l_1'|l_2|l_2'|B_3|\dots |B_r}-\x_{l_1|l_1'|B_2|B_3|\dots |B_r}-\x_{B_1|l_2|l_2'|B_3|\dots |B_r}.$$

Note that all partitions involved in this expression, except for $\sigma$, have size larger than $|\sigma|$
and we can apply the induction hypothesis to any $\x_{\sigma'}$ appearing here with $\sigma'$ incompatible, to write $\x_\sigma$ as a linear
combination of $\x_{\sigma'}'s$ using only compatible $\sigma'$.

If $\sigma$ has singletons, we remove these singletons in $T$ and $\sigma$ obtaining a tree $T_0$ and a partition $\sigma_0$ without singletons
on $T_0$. We apply the previous argument to $\sigma_0$ and $T_0$ to obtain a linear invariant.
Then we apply the Extension Lemma \ref{lem_ext}(a) recursively to add singletons and we end up also with a linear polynomial
that involves $\sigma$ and partitions of larger size. Hence, we can apply the induction hypothesis again.

The linear invariants obtained in this way for each incompatible partition $\sigma$ are of Lake--type and form a set of
linearly independent vectors in $L_T/L$ because they involve partitions of larger size.
\end{proof}

\begin{remark}\rm
\textbf{Case $\kappa=2$.} For $\kappa=2$, Theorem \ref{thm_dimension} and Corollary \ref{cor_Lake} do not apply.
In this case it is already known (see \citet{mat08}) that there are no  linear topology invariants for the uniform distribution
$\pi_U$ and hence $\mathcal{D}^{\pi_U}_T=\mathcal{D}^{\pi_U}$ for any tree $T$ (see \citet{mat08}).
One can actually prove that this also holds for any generic distribution $\pi$ and this space has dimension  $|\Sigma_2|= 2^{n-1}-1$, see \cite{mat08}.
\end{remark}

\begin{remark}\rm
\textbf{Case $\kappa=3$.} For $\kappa=3$ and $n=4$, we cannot apply Corollary \ref{cor_Lake} either. But in this case we can provide another topology invariant. We describe it in the following lemma
for $n=4$ but can be easily generalized for the uniform distribution to any tree by using a similar argument as in Proposition \ref{lake}.
Moreover, it is  not difficult to see that for $\kappa\geq 4$ it can be derived form Lake--type invariants.
\end{remark}

\begin{lemma}\label{lem_3states}
For the tree $12|34$ and any positive distribution $\pi$ on a set $S$ of  $\kappa \geq 3$ states, the polynomial
\begin{equation}\label{triple}
\tx_{xyxy}+\tx_{xyyz}+\tx_{xyzx}-\tx_{xyyx}-\tx_{xyxz}-\tx_{xyzy},
\end{equation}
for any three different states $x,y,z \in S$,
is a topology invariant if $T$ evolves under the EI/RC model.
\end{lemma}

\begin{proof} According to Lemma \ref{lem_forests} we need to prove that  \eqref{triple} vanishes when we evaluate it at the points
$q_F$, $F \in \mathcal{F}_T$. If $F$ is a forest such that $\sigma(F)$ does not refine any of the partitions $\{\{1,3\},\{2,4\}\}$,
$\{\{1,4\},\{2,3\}\}$, 
then the coordinates that appear in \eqref{triple} are all zero. If $\sigma(F)$ refines $\{\{1,3\},\{2,4\}\}$,
then $\sigma(F)$ is either $\{\{1,3\},\{2\},\{4\}\}$, or $\{\{2,4\},\{1\},\{3\}\}$ or the trivial forest.
In the first two cases \eqref{triple} evaluated at $q_F$ vanishes.
As the evaluation of any coordinate $\tx$ at the point associated to the trivial forest is one,
it also vanishes in this case. The remaining cases follow from the symmetry of leaves $3$ and $4$ in \eqref{triple}.
\end{proof}

\begin{remark}\rm
\textbf{Case $\kappa=4$.} For $n=5$ not all linear topology invariants are of Lake--type. In \citet{fu95} a complete list of 17 ($=|\Sigma_4|-|co(T)|=61-34$)
linear invariants that  generate  the space of linear topology invariants is given.
For example, for the fully  symmetric model on the set of states
$\{x,y,z,w\}$ (i.e. Jukes-Cantor model),
$$\x_{xyyxy}+\x_{xyzwz}-\x_{xyyzy}-\x_{xyzxz}$$
is a topology linear invariant that cannot be described by Proposition  \ref{lake}.
\end{remark}

\section{Explicit linear invariants for quartet trees}
\label{Explicit_linear_invariants_for_quartet_trees}
In this section we assume that $\kappa\geq 4$ and we shall deal with quartet trees and the star tree on four leaves. Note that in the previous section we gave an explicit description of linear phylogenetic invariants only when the distribution was uniform. For a generic distribution $\pi$ we managed to compute the dimension of the space of linear phylogenetic invariants, but we did not provide a explicit set of generators. We do it in this section for the case $n=4$, $\kappa\geq 4$, and any distribution $\pi$.

\begin{remark}\label{rmk_modelinvar}\rm
In the case of quartet trees on the set of taxa $X=[4]$, the possible tree topologies are $12|34$, $13|24$, $14|23$, and the star tree $T_*$. As the star tree is a subtree of the others, the vector space of phylogenetic mixtures is
$$E^{\pi}=\langle  \vec{q}_F \mid F \in  \mathcal{F}_{12|34} \rangle +\langle  \vec{q}_F \mid F \in  \mathcal{F}_{12|34} \rangle+\langle  \vec{q}_F \mid F \in  \mathcal{F}_{13|24} \rangle +\langle q_F \mid F \in  \mathcal{F}_{14|23} \rangle .$$
By Theorem \ref{thm_dimension} we know that the vectors $ \vec{q}_F$ are linearly independent if we let  $F$ move in the set of full subforests of the tree $A|B$. As $\mathcal{F}_{12|34}, \mathcal{F}_{13|24} $ and $\mathcal{F}_{14|23}$ intersect at the set of forests for the star tree $T_*$, in order to check whether a homogeneous linear polynomial vanishes at the vectors of $E^{\pi}$ one needs to check whether it vanishes at the 15 vectors of Tables 1, 2 and 3 that correspond to $12$ subforests of $T_*$ and one forest $ \vec{q}_{A|B}$ for each refined quartet).
\end{remark}

\begin{proposition}\label{prop_n4} Let $x,y,z,w$ be four different states and define
$$\beta_{x,y}=\pi_x^2\tx_{xxxy}+\pi_x\pi_y[\tx_{xxyy}+\tx_{xyxy}+\tx_{yxxy}]+\pi_x\pi_y[\tx_{zwxx}+\tx_{zxwx}+\tx_{xzwx}] +\pi_y^2\tx_{xyzw} ,$$
\begin{multline*}
\delta_{x,y}=\pi_x^2[\pi_x\tx_{xxxx}+\pi_y\tx_{xxxy}+\pi_z\tx_{xxxz}+\pi_w\tx_{xxxw}]+\\
+\pi_x\pi_y[\pi_x\tx_{xxyx}+\pi_y\tx_{xxyy}+\pi_z\tx_{xxyz}+\pi_w\tx_{xxyw}]+\\
+\pi_x\pi_y[\pi_x\tx_{xyxx}+\pi_y\tx_{xyxy}+\pi_z\tx_{xyxz}+\pi_w\tx_{xyxw}]+\\
+\pi_x\pi_y[\pi_x\tx_{yxxx}+\pi_y\tx_{yxxy}+\pi_z\tx_{yxxz}+\pi_w\tx_{yxxw}]+\\
+\pi_y^2[\pi_x\tx_{xyzx}+\pi_y\tx_{xyzy}+\pi_z\tx_{xyzz}+\pi_w\tx_{xyzw}].
\end{multline*}

Then following are linear model invariants for  quartet trees evolving under the $EI/RC$ model:
\begin{eqnarray}
\pi_y \tx_{xxyy}+\pi_z\tx_{xxyz}&-&\pi_y\tx_{xxzy} +\pi_z \tx_{xxzz} \label{eq_xxyy}\\
\pi_x \tx_{xxyz}+\pi_{w}\tx_{xwyz}&-&\pi_{w}\tx_{wwyz} +\pi_x \tx_{wxyz} \label{eq_xxyz}\\
\beta_{x,y}&-& \beta_{y,x} \label{eq_beta}\\
\delta_{x,y}&-& \delta_{y,x} \label{eq_delta}
\end{eqnarray}
One obtains analogous linear model invariants by considering any permutation of the set of leaves.

\end{proposition}

\begin{proof}
From the extension Lemma \ref{lem_ext}(b) it follows that \eqref{eq_xxyy} and \eqref{eq_xxyz} are model invariants.
Indeed, if we consider the star tree $T_2$ on two leaves,  then it is easy to check that
$$\pi_y\tx_{yy}+\pi_{z}\tx_{yz}-\pi_y\tx_{zy}-\pi_z\tx_{zz}$$
is a linear phylogenetic invariant. By identifying $T_2$ with the star tree  $T_{3,4}$ on leaves $3,4$ we can apply
Lemma \ref{lem_ext}(b) with $\mu=xx$ to obtain \eqref{eq_xxyy} for the quartet tree $T=12|34$ (because $T$ can be obtained by
attaching the tripod tree $T_{1,2,l}$ to the edge leading to leaf $3$ of $T_2$). 
In particular, $\eqref{eq_xxyy}$ vanishes for the star tree $T_*$ on four leaves.
 Similarly, in order to see that \eqref{eq_xxyz} is a phylogenetic invariant for the star tree $T_*$, we use the phylogenetic invariant
$$\pi_x\tx_{xx}+\pi_{w}\tx_{xw}-\pi_w\tx_{ww}-\pi_x\tx_{wx}$$
for the tree $T_2=T_{1,2}$ and apply Lemma \ref{lem_ext}(b) with $\mu=yz$.
By Lemma \ref{lem_ext}(c) we see that \eqref{eq_xxyz} is a phylogenetic invariant for the quartet tree $12|34$ (and hence also for the star tree  $T_*$).

In order to prove that \eqref{eq_xxyy} and \eqref{eq_xxyz} are model invariants, it only remains to check that these expression vanish when evaluated at $ \vec{q}_{13|24}$ and $ \vec{q}_{14|23}$,
which is straight forward because all coordinates involved in the expressions are 0 for these vectors.


We check now that $\eqref{eq_beta}$ and $\eqref{eq_delta}$ are model invariants having Remark \ref{rmk_modelinvar} in mind.
Looking at Table 1, we observe that $\beta_{x,y}$ (respectively $\delta_{x,y}$)  evaluated at $ \vec{q}_{\bullet}$ is $\pi_x^2+6\pi_x\pi_y+\pi_y^2$ (resp. $\pi_x^2+3\pi_x\pi_y+\pi_y^2(\pi_x+\pi_y+\pi_z+\pi_w)$).
As these expressions are symmetric for $x$ and $y$, \eqref{eq_beta} and $\eqref{eq_delta}$ vanish in this case.

Now we consider the other vectors in Table 1, $ \vec{q}_{B}$, where $B$ is a block of $m$ leaves, $m\geq 2$, and the partition associated to this point is $B$ and singleton blocks.

\noindent
We start with $m=2$. Using the equalities of lemma \ref{lem_equal}, we can see that $\beta_{x,y}$ and $\delta_{x,y}$ are symmetric under the permutation of leaves 1,2, and  3.
Thus we only need to consider that  $B$ is formed either by  $\{1,2\}$ or by $\{3,4\}$. In the first case,
$\beta_{x,y}$ evaluated at $ \vec{q}_B$ is $\pi_x+\pi_y$ and $\delta_{x,y}$ is $(\pi_x+\pi_y)(\pi_x+\pi_y+\pi_z+\pi_w)$. As these expressions are symmetric in $x$ and $y$, \eqref{eq_beta} and \eqref{eq_delta} also vanish in this case.
If $B=\{3,4\}$, then
the evaluation of $\beta_{x,y}$ at $ \vec{q}_B$ equals
  $\pi_x+\pi_y$ and the evaluation of \eqref{eq_delta} gives $\pi_x^2+3\pi_x\pi_y+\pi_y^2$. Again, these are symmetric in $x,y$ and \eqref{eq_beta}, \eqref{eq_delta} vanish.


Now we consider $m=3$. Let us assume first that $B=\{1,2,3\}$.  In this case,
the evaluation of $\beta_{x,y}$ at $ \vec{q}_B$ equals $1$ and the evaluation of $\delta{x,y}$ is $\pi_x+\pi_y+\pi_z+\pi_w$. Therefore \eqref{eq_beta} and \eqref{eq_delta} vanish at $ \vec{q}_B$.
If $B$ contains the leaf $4$, then all terms in the evaluation of  $\beta_{x,y}$ at $ \vec{q}_B$ are zero and the evaluation of $\delta{x,y}$ at $ \vec{q}_B$ is $\pi_x+\pi_y$. Therefore \eqref{eq_beta} and $\eqref{eq_delta}$ also hold for these vectors.

If $m=4$, then \eqref{eq_beta} vanishes trivially because all its terms are 0. Moreover $\delta{x,y}$ is equal to 1 when evaluated at $ \vec{q}_{1234}$ and there fore both equations hold for this vector.

The only remaining cases to check correspond to the vectors $ \vec{q}_{12|34}$, $ \vec{q}_{13|24}$ and $ \vec{q}_{14|23}$ of Tables 2 and 3.
As $\beta_{x,y}$ is equal to 1 and $\delta_{x,y}$ is equal to $\pi_x+\pi_y$ when these expressions are evaluated at these vectors, both equations \eqref{eq_beta} and \eqref{eq_delta} vanish on these vectors.
%
%

Note that when we apply a permutation of the set of leaves, the resulting polynomials are phylogenetic invariants because we have just proven that the original ones are  linear model invariants.
\end{proof}

\begin{theorem}\label{prop_linspace} For any distribution $\pi$, the space of linear model invariants $L^{\pi}$ for $n=4$ and $\kappa\geq 4$
is generated by the phylogenetic invariants of Proposition \ref{prop_n4} together with  $\tx_{\chi}-\tx_{\chi'}$ for any $\chi \equiv \chi'$
and has dimension $\kappa^4-B_4=\kappa^4-15$.
\end{theorem}

For the fully symmetric model we have already seen in Remark \ref{rem_symmetric} that $\x_{\chi}-\x_{\chi'}$ are linear phylogenetic invariants
if $\sigma(\chi) =\sigma(\chi')$. In this case this set of invariants defines the same
vector space as the phylogenetic invariants in Theorem \ref{prop_linspace}.

\begin{remark}\rm Although one could replace \eqref{eq_beta}  by other phylogenetic invariants obtained from marginalization
from a phylogenetic invariant relating $\tx_{xxy}$ and $\tx_{yyx}$ on the tripod, this expression would have less symmetries than
\eqref{eq_beta} and therefore we decided to use \eqref{eq_beta} instead (similarly for \eqref{eq_delta}).
\end{remark}

\begin{proof}
We let  $F^{\pi}$  be the space of vectors where all the linear polynomials in the statement vanish.
Then we shall prove that for the vectors in $F^{\pi}$, any coordinate $\tx_{\chi}$ can be expressed as a linear combination of the
following 15 coordinates:
\begin{gather*}
\tx_{xxxx}\\
\tx_{xxxy},\tx_{xxyx},\tx_{xyxx},\tx_{yxxx}\\
\tx_{ xxyy},\tx_{xyxy},\tx_{xyyx}\\
\tx_{ xxyz},\tx_{xyxz},\tx_{xyzx},\tx_{yxzx},\tx_{yxxz},\tx_{yzxx}\\
 \tx_{xyzw}
\end{gather*}

This will prove that $F^{\pi}$ is a vector space of dimension 15 or lower.
By Lemma \ref{lem_pointsmodel} we know that $\dim \mathcal{D}^{\pi}$ is $\geq |\Sigma_{\kappa}|-1$, which is $B_4-1=14$ for $n=4$. As we have
the inclusion $ \mathcal{D}^{\pi} =E^{\pi}\cap H \subseteq F^{\pi}\cap H$ this will finish the proof.

First note that by Lemma \ref{lem_equal} we have $\tx_{xxxy'}=\tx_{xxxy}$, $\tx_{xxy'z'}=\tx_{xxyz}$, $\tx_{x'y'z'w'}=\tx_{xyzw}$ for any $y'\neq x,x'$, $z'\neq y,y',x,x'$, $w'\neq x,y,z,x',y',z'$.

Using the equation \eqref{eq_xxyz}=0 one can write $\tx_{x'x'y'z'}$ as a linear combination of $\tx_{xxyz}$ and $\tx_{xyzw}$.
The equation \eqref{eq_xxyy}=0 allows us to put $\tx_{xxy'y'}$  as a linear combination of $\tx_{xxyy}$ if $y'\neq y$.
In order to write $\tx_{yyxx}$ (or similarly $\tx_{yxxy}$) in terms of the allowed coordinates we need to do two steps. We use expression $\eqref{eq_xxyy}$ three times to put first $\tx_{yyxx}$ in terms of $\tx_{yyzz}$ first, then $\tx_{yyzz}$ in terms of $\tx_{xxzz}$ and finally $\tx_{xxzz}$ in terms of $\tx_{xxyy}$.
Interchanging the role of leaves 1,2 with  3,4 we also obtain $\tx_{x'x'yy}$ as a linear combination of $\tx_{xxyy}$ if $x'\neq x$.
In the same way, we can use the equation \eqref{eq_beta}=0 to put $\tx_{x'x'x'y'}$ as a linear  combination of $\tx_{xxxy}$ and other coordinates
which we now know that are linear combinations of the allowed coordinates. 
Finally, we use the equation \eqref{eq_delta}=0 to put $\tx_{x'x'x'x'}$ for $x'\neq x$ as a linear combination of $\tx_{xxxx}$ and other allowed coordinates.

By considering these relations above and all permutations of the leaves, we end up with every coordinate written as a linear combination of  the allowed list of 15 coordinates.
\end{proof}

We now consider the two linear topology invariants
that we obtained in Example~\ref{ex_lake}: in terms of the $\tx's$ above, the corresponding equations for the quartet tree $12|34$ these are
$$H_1: \quad \tx_{xyxy}+\tx_{xyzw}=\tx_{xyzy}+\tx_{xyxw}$$
$$H_2: \quad \tx_{xyyx}+\tx_{xywz}=\tx_{xyyz}+\tx_{xywx}.$$

Equations $H_1$ and $H_2$ are linearly independent and drop the dimension by two.
In total, we have that  $\mathcal{D}_{12|34}^{\pi}$ is contained in an affine space $E^{\pi}\cap H  \cap H_1\cap H_2$ of dimension 12. As the dimension
of $\mathcal{D}_{12|34}^{\pi}$ is 12 and for the star tree $\dim \mathcal{D}_{T_*}^{\pi}=11$  we have:

\begin{corollary}
\label{corocor}
For $n=4$ and any distribution $\pi$ one has
$$\mathcal{D}^{\pi}=E^{\pi}\cap H$$
$$\mathcal{D}_{12|34}^{\pi}=E^{\pi}\cap H\cap H_1\cap H_2$$
$$\mathcal{D}_{T_*}^{\pi}=E^{\pi}\cap H\cap H_1\cap H_2\cap H_3$$
where $H_3:$ $\tx_{xxyy}+\tx_{xzyw}=\tx_{xzyy}+\tx_{xxyw}$ and ${T_*}$ denotes the star tree on four leaves. In particular, Lake--type invariants generate all linear topology invariants for quartet trees evolving under the $EI$ model.
\end{corollary}

\section{The infinite-state random cluster model $RC_{\infty}$}
\label{The_infinite-state_random_cluster_model}
Recall that in the random cluster model, each edge of $T$ is cut with some probability $\theta_e$ to obtain a resulting partition $\sigma$ of
the leaf set $X$. Each block is then assigned a state independently according to the distribution $\pi$.
However, we could just consider the partition $\sigma$ itself as the output of this process
(rather than assigning states, which has the effect of combining some blocks together when they receive the same state).
We call this the {\em infinite state $RC$ model} $RC_{\infty}$ since it has a natural interpretation as the limiting distribution on partitions induced by the $EI/RC$ model as the number of states $\kappa$ in $S$ tends to infinity when states have at least
roughly similar probabilities.

More precisely, under the $RC$ model, the probability that two blocks of $\sigma$ are assigned a same state in the equal input model is at
most $n\sum_{\alpha \in S} \pi_\alpha^2$, by Boole's inequality (note that there are at most
$n$ blocks in $\sigma$).  Suppose that  $\pi_\alpha \in [a/k, b/k]$ for some fixed $a,b$ then as $k=|S| \rightarrow \infty$ all blocks of
$\sigma$ receive distinct states with probability converging to 1 (this restriction on $\pi$ can be weakened a little further).
The $RC_{\infty}$ model is sometimes referred to as the `Kimura's infinite alleles' model in phylogenetics, and it was studied mathematically in \citet{mos04}.

\subsection{Linear invariants for $RC_{\infty}$}
The linear phylogenetic invariants for the infinite-state random cluster model are particularly easy to describe.


Let $p_\sigma = \PP_T(\sigma | \Theta)$ be the probability of generating partition $\sigma$ on $T$ under the
$RC_\infty$ model with edge cut probabilities $\Theta = (\theta_e)$, and recall the definitions of ${\co}(T)$ and ${\Inc}(T)$ from Section ~\ref{Combinatorial_concepts_and_terminology}.

\begin{proposition}
\label{proinf}
Under the $RC_\infty$ model:
\begin{itemize}
\item[(i)]  $\PP_T(\sigma | \Theta) =0$ for all $\Theta$ if and only if $\sigma \in \Inc(T)$.
\item[(ii)]  $\{\x_\sigma : \sigma \in {\Inc}(T)\}$ forms a basis for the vector space $L_T$ of
linear phylogenetic invariants for $T$ and of the space of linear topology invariants.
Consequently,  this space has dimension $|{\Inc}(T)| = B_n - |{\co}(T)|$.
\item[(iii)] The space of all phylogenetic mixtures on $T$ has dimension $|{\ co}(T)|-1$.
\item[(iv)] The space of all phylogenetic mixtures on all $n$--leaf trees under the $RC_\infty$ model has dimension $B_n -1$.
\end{itemize}
\end{proposition}
\begin{proof}
(i) Suppose that $\sigma \in \Inc(T)$. Then there exists two blocks $B, B'$ of $\sigma$ and leaves $x,y \in B$ and
$x', y' \in B'$ for which the paths $P(T; x,y)$ and $P(T; x',y')$ share at least one vertex.
Now since $x,y \in B$ and $x',y' \in B'$ the only way to generate $\sigma$ under $RC_\infty$ is if none
of the edges in the two paths $P(T; x,y)$ and $P(T; x',y')$ is cut. Since these paths intersect on a vertex this implies that $x$ and $x'$ must be
the same block, i.e. that $B = B'$. Thus $\sigma$ cannot be generated with positive probability under the $RC_\infty$ model.
Conversely, suppose that $\sigma$ is convex on $T$. Then set $\theta_e = 0$ for
all edges in $\{T[B]: B \in \sigma\}$ and set $\theta_e=1$ for all other edges. Then $p_\sigma = 1$.

\noindent
(ii)   If $\sum\lambda_{\sigma}\x_{\sigma}$ is a linear phylogenetic invariant, then for any $\sigma$
convex on $T$ we can choose a set of parameters $\Theta$ such that $p_{\sigma}=1$ (see above).
This implies that $\lambda_{\sigma}=0$ for any $\sigma \in {\co}(T)$.
This and (i) show that the set spans the space of all linear phylogenetic invariants, and linear independence
follows immediately from the observation that each polynomial involves a variable not present in any other polynomial in this set.
Note that all these polynomials are topology invariants.

\noindent
(iii) The space of phylogenetic mixtures $\mathcal{D}_T$ on $T$ is equal to $E_T \cap H$ where $E_T$ is the space of vectors on which the
linear phylogenetic invariants vanish and $H$ is the hyperplane defined by the trivial equation $\sum_{\sigma} \x_{\sigma}=1$
(the sum is over all partitions of $[n]$).  By (ii), $E_T$ has dimension $B_n-{\Inc}(T)=|{\co}(T)|$ and we are done.

\noindent
(iv) Note that in the basis $\{\x_\sigma : \sigma \in {\rm Inc}(T)\}$ of (ii) there are no model invariants.
Therefore, the set $\mathcal{D}$ of phylogenetic mixtures on all trees coincides with the trivial hyperplane $H$ and has dimension $B_n-1$.
\end{proof}

The construction of  certain quadratic phylogenetic invariants for $RC_\infty$ is also quite easy.
Let $x\sim y$ denote the event that $x$ and $y$ are in the same block of the partition generated by a phylogeny under the $RC_\infty$ model, and let $p(x,y)$ denote the probability of that event.  Note that $p(x,y)$ is a sum of
$p_\sigma$ values over all $\sigma$ for which $x$ and $y$ are in the same block.
Then $p(x,y) = \prod_{e \in P(T; x,y)} (1-\theta_e)$, where $P(T; x,y)$ is the path in $T$ between $x$ and $y$.
It follows (from the four point condition) that if the quartet tree obtained by restricting $T$ to $x,y,w,z$ is either
$xy|wz$ or the star tree, then
$$p(x,w)p(y,z)-p(x,z)p(y,w)=0.$$

\section{Future work}
\label{Future_work}
It would be interesting to generalize Lake--type invariants in such a way that they generate the space of linear topology invariants
for $\kappa< n$ ({\em cf.} Corollary \ref{cor_Lake}).
On the other hand, it also would be useful to give explicit linear model invariants (with many symmetries) for any number of leaves,
as  was done in Section 4 for $n=3,4$. These model invariants could be used for model selection as it was done in \citet{ked12}
for the uniform distribution. Extending the work of Section 4 to other models is also of interest because this would increase the range of
models that can be considered in certain model selection software such as SPIn (http://genome.crg.es/cgi-bin/phylo\_mod\_sel/AlgModelSelection.pl).


\section*{Acknowledgements}
We thank the two anonymous reviewers for their helpful comments on an earlier version of this manuscript.
Part of this research was performed while MC was visiting the Biomathematics Research Center of the University of Canterbury.
MC would like to thank the Biomathematics Research Center (and specially its director) for the invitation, the support provided, and the great working atmosphere. MC is partially supported by MTM2012-38122-C03-01, MTM2015-69135-P (MINECO/FEDER) and Generalitat de Catalunya 2014 SGR-634.

\bibliographystyle{jtbnew}

\begin{thebibliography}{19}
\expandafter\ifx\csname natexlab\endcsname\relax\def\natexlab#1{#1}\fi
\providecommand{\url}[1]{\texttt{#1}}
\providecommand{\href}[2]{#2}
\providecommand{\path}[1]{#1}
\providecommand{\DOIprefix}{doi:}
\providecommand{\ArXivprefix}{arXiv:}
\providecommand{\URLprefix}{URL: }
\providecommand{\Pubmedprefix}{pmid:}
\providecommand{\doi}[1]{\href{http://dx.doi.org/#1}{\path{#1}}}
\providecommand{\Pubmed}[1]{\href{pmid:#1}{\path{#1}}}
\providecommand{\bibinfo}[2]{#2}
\ifx\xfnm\relax \def\xfnm[#1]{\unskip,\space#1}\fi
\bibitem[{Allman et~al.(2012)Allman, Rhodes and Sullivant}]{all12}
\bibinfo{author}{Allman, E.S.}, \bibinfo{author}{Rhodes, J.A.},
  \bibinfo{author}{Sullivant, S.}, \bibinfo{year}{2012}.
\newblock \bibinfo{title}{When do phylogenetic mixture models mimic other
  phylogenetic models?}
\newblock \bibinfo{journal}{Syst. Biol.} \bibinfo{volume}{61},
  \bibinfo{pages}{1049--1059}.
\bibitem[{Casanellas and Fern{\'a}ndez-S{\'a}nchez(2011)}]{cas11}
\bibinfo{author}{Casanellas, M.}, \bibinfo{author}{Fern{\'a}ndez-S{\'a}nchez,
  J.}, \bibinfo{year}{2011}.
\newblock \bibinfo{title}{Relevant phylogenetic invariants of evolutionary
  models}.
\newblock \bibinfo{journal}{J. Math. Pures Appl.} \bibinfo{volume}{96},
  \bibinfo{pages}{207--229}.
\bibitem[{Casanellas et~al.(2012)Casanellas, Fern{\'a}ndez-S{\'a}nchez and
  Kedzierska}]{cas12}
\bibinfo{author}{Casanellas, M.}, \bibinfo{author}{Fern{\'a}ndez-S{\'a}nchez,
  J.}, \bibinfo{author}{Kedzierska, A.M.}, \bibinfo{year}{2012}.
\newblock \bibinfo{title}{The space of phylogenetic mixtures for equivariant
  models}.
\newblock \bibinfo{journal}{Alg. Mol. Biol.} \bibinfo{volume}{7},
  \bibinfo{pages}{33}.
\bibitem[{Chang(1996)}]{cha96}
\bibinfo{author}{Chang, J.T.}, \bibinfo{year}{1996}.
\newblock \bibinfo{title}{Full reconstruction of {M}arkov models on
  evolutionary trees: Identifiability and consistency}.
\newblock \bibinfo{journal}{Math. Biosci.} \bibinfo{volume}{137},
  \bibinfo{pages}{51--73}.
\bibitem[{Felsenstein(2004)}]{fel04}
\bibinfo{author}{Felsenstein, J.}, \bibinfo{year}{2004}.
\newblock \bibinfo{title}{Inferring Phylogenies}.
\newblock \bibinfo{publisher}{Sinauer Associates},
  \bibinfo{address}{Sunderland, MA}.
\bibitem[{Fern{\'a}ndez-S{\'a}nchez and Casanellas(2016)}]{fer15}
\bibinfo{author}{Fern{\'a}ndez-S{\'a}nchez, J.}, \bibinfo{author}{Casanellas,
  M.}, \bibinfo{year}{2016}.
\newblock \bibinfo{title}{Invariant versus classical quartet inference when
  evolution is heterogeneous across sites and lineages}.
\newblock \bibinfo{journal}{Syst. Biol.} \bibinfo{volume}{65},
  \bibinfo{pages}{280--291}.
\bibitem[{Fu(1995)}]{fu95}
\bibinfo{author}{Fu, Y.X.}, \bibinfo{year}{1995}.
\newblock \bibinfo{title}{Linear invariants under {J}ukes' and {C}antor's
  one-parameter model}.
\newblock \bibinfo{journal}{J. Theor. Biol.} \bibinfo{volume}{173},
  \bibinfo{pages}{339--352}.
\bibitem[{Fu and Li(1991)}]{fu91}
\bibinfo{author}{Fu, Y.X.}, \bibinfo{author}{Li, W.H.}, \bibinfo{year}{1991}.
\newblock \bibinfo{title}{Necessary and sufficient conditions for the existence
  of certain quadratic invariants under a phylogenetic tree}.
\newblock \bibinfo{journal}{Math. Biosci.} \bibinfo{volume}{105},
  \bibinfo{pages}{229--238}.
\bibitem[{Kedzierska et~al.(2012)Kedzierska, Drton, Guig\'o and
  Casanellas}]{ked12}
\bibinfo{author}{Kedzierska, A.}, \bibinfo{author}{Drton, M.},
  \bibinfo{author}{Guig\'o, R.}, \bibinfo{author}{Casanellas, M.},
  \bibinfo{year}{2012}.
\newblock \bibinfo{title}{{SPI}n: model selection for phylogenetic mixtures via
  linear invariants}.
\newblock \bibinfo{journal}{Mol. Biol. Evol.} \bibinfo{volume}{29},
  \bibinfo{pages}{929--937}.
\bibitem[{Kemeny and Snell(1976)}]{kem76}
\bibinfo{author}{Kemeny, J.G.}, \bibinfo{author}{Snell, J.L.},
  \bibinfo{year}{1976}.
\newblock \bibinfo{title}{Finite {M}arkov chains}.
\newblock \bibinfo{publisher}{Springer-Verlag}, \bibinfo{address}{New York}.
\bibitem[{Lake(1987)}]{lak87}
\bibinfo{author}{Lake, J.}, \bibinfo{year}{1987}.
\newblock \bibinfo{title}{A rate-independent technique for analysis of nucleic
  acid sequences: evolutionary parsimony}.
\newblock \bibinfo{journal}{Molec. Biol. Evol.} \bibinfo{volume}{4},
  \bibinfo{pages}{167--191}.
\bibitem[{Matsen et~al.(2008)Matsen, Mossel and Steel}]{mat08}
\bibinfo{author}{Matsen, F.A.}, \bibinfo{author}{Mossel, E.},
  \bibinfo{author}{Steel, M.}, \bibinfo{year}{2008}.
\newblock \bibinfo{title}{Mixed-up trees: The structure of phylogenetic
  mixtures}.
\newblock \bibinfo{journal}{Bull. Math. Biol.} \bibinfo{volume}{70},
  \bibinfo{pages}{1115--1139}.
\bibitem[{Mossel and Steel(2004)}]{mos04}
\bibinfo{author}{Mossel, E.}, \bibinfo{author}{Steel, M.},
  \bibinfo{year}{2004}.
\newblock \bibinfo{title}{A phase transition for a random cluster model on
  phylogenetic trees}.
\newblock \bibinfo{journal}{Math. Biosci.} \bibinfo{volume}{187},
  \bibinfo{pages}{189--203}.
\bibitem[{Semple and Steel(2003)}]{sem03}
\bibinfo{author}{Semple, C.}, \bibinfo{author}{Steel, M.},
  \bibinfo{year}{2003}.
\newblock \bibinfo{title}{Phylogenetics}.
\newblock \bibinfo{publisher}{Oxford University Press}.
\bibitem[{Steel(2011)}]{ste11}
\bibinfo{author}{Steel, M.}, \bibinfo{year}{2011}.
\newblock \bibinfo{title}{Can we avoid `{SIN}' in the house of `no common
  mechanism'?}
\newblock \bibinfo{journal}{Syst. Biol.} \bibinfo{volume}{60},
  \bibinfo{pages}{96--109}.
\bibitem[{Steel and Fu(1995)}]{ste95}
\bibinfo{author}{Steel, M.A.}, \bibinfo{author}{Fu, Y.X.},
  \bibinfo{year}{1995}.
\newblock \bibinfo{title}{Classifying and counting linear phylogenetic
  invariants for the {J}ukes--{C}antor model}.
\newblock \bibinfo{journal}{J. Comput. Biol.} \bibinfo{volume}{2},
  \bibinfo{pages}{39--47}.
\bibitem[{Steel et~al.(1994)Steel, Sz{\'e}kely and Hendy}]{ste94}
\bibinfo{author}{Steel, M.A.}, \bibinfo{author}{Sz{\'e}kely, L.A.},
  \bibinfo{author}{Hendy, M.D.}, \bibinfo{year}{1994}.
\newblock \bibinfo{title}{Reconstructing trees when sequence sites evolve at
  variable rates}.
\newblock \bibinfo{journal}{J. Comput. Biol.} \bibinfo{volume}{1},
  \bibinfo{pages}{153--163}.
\bibitem[{Sturmfels and Sullivant(2005)}]{sturmfels2004}
\bibinfo{author}{Sturmfels, B.}, \bibinfo{author}{Sullivant, S.},
  \bibinfo{year}{2005}.
\newblock \bibinfo{title}{Toric ideals of phylogenetic invariants}.
\newblock \bibinfo{journal}{J. Comput. Biol.} \bibinfo{volume}{12},
  \bibinfo{pages}{204--228}.
\bibitem[{\v{S}tefakovi\v{c} and Vigoda(2007)}]{stef07}
\bibinfo{author}{\v{S}tefakovi\v{c}, D.}, \bibinfo{author}{Vigoda, E.},
  \bibinfo{year}{2007}.
\newblock \bibinfo{title}{Phylogeny of mixture models: robustness of maximum
  likelihood and non-identifiable distributions}.
\newblock \bibinfo{journal}{J. Comput. Biol.} \bibinfo{volume}{14},
  \bibinfo{pages}{156--189}.

\end{thebibliography}

\end{document}